\documentclass[12pt, hidelinks]{article}

\usepackage{mypreample}
\usepackage{authblk}
\usepackage[twoside=true, margin=1in]{geometry}

\author[1]{Matheus V. X. Ferreira\thanks{Contact: mvxf@cs.princeton.edu}\thanks{Part of this work was completed while visiting Harvard University.}}
\author[2]{Daniel J. Moroz}
\author[2]{David C. Parkes}
\author[3]{Mitchell Stern\protect\footnotemark[3]}
\affil[1]{Computer Science, Princeton University}
\affil[2]{Computer Science, Harvard University}
\affil[3]{EECS, University of California, Berkeley}

\title{Dynamic Posted-Price Mechanisms for the Blockchain Transaction Fee Market\thanks{The authors thank Cemil Dibek, Vikram V. Ramaswamy, and Tim Roughgarden for helpful discussions.}}

\begin{document}

\begin{titlepage}
\maketitle

\begin{abstract}
In recent years, prominent blockchain systems such as Bitcoin and Ethereum have experienced explosive growth in transaction volume, leading to frequent surges in demand for limited block space and causing transaction fees to fluctuate by orders of magnitude. Existing systems sell space using first-price auctions~\cite{nakamoto2008bitcoin}; however, users find it difficult to estimate how much they need to bid in order to get their transactions accepted onto the chain. If they bid too low, their transactions can have long confirmation times. If they bid too high, they pay larger fees than necessary.

In light of these issues, new transaction fee mechanisms have been proposed, most notably EIP-1559~\cite{buterin2019eip1559}, aiming to provide better usability. EIP-1559 is a history-dependent mechanism that relies on block utilization to adjust a base fee. We propose an alternative design---a {\em dynamic posted-price mechanism}---which uses not only block utilization but also observable bids from past blocks to compute a posted price for subsequent blocks. We show its potential to reduce price volatility by providing examples for which the prices of EIP-1559 are unstable while the prices of the proposed mechanism are stable. More generally, whenever the demand for the blockchain stabilizes, we ask if our mechanism is able to converge to a stable state. Our main result provides sufficient conditions in a probabilistic setting for which the proposed mechanism is approximately welfare optimal and the prices are stable. Our main technical contribution towards establishing stability is an iterative algorithm that, given oracle access to a Lipschitz continuous and strictly concave function $f$, converges to a fixed point of $f$.
\end{abstract}

\end{titlepage}

\tableofcontents

\section{Introduction}

Due to the explosive interest in blockchains such as Ethereum and Bitcoin, allocating scarce block space has become a significant problem. The high competition for block space has made transaction fee auctions a multibillion-dollar venture for miners~\cite{graffeo}. The status quo, first adopted by Bitcoin~\cite{nakamoto2008bitcoin}, is to auction block space via a first-price auction where miners select transactions with the highest fees to be included in a block.

Since the demand for block space varies with time, users need to continually change their bids to maximize their chances of getting a transaction included and, at the same time, minimize the amount they pay. This challenge has resulted in the proposal of new  mechanisms~\cite{lavi2019redesigning, basu2019towards, buterin2019eip1559} that pay special attention to the benefits of mechanisms that are {\em incentive compatible} (or strategy-proof, promoting truthful bids). A particular challenge for designing transaction fee auctions is that the miners are also auctioneers. While the blockchain can commit to implementing a payment rule, miners have complete control over the transactions included in a block (the {\em winner determination}). Moreover, there is information asymmetry, where the blockchain only sees the transactions that are included in blocks and not, for example, the excess demand.

Designing strategy-proof auctions that are at the same time robust to strategic auctioneers is a challenging task and closely related to the problem of designing credible strategy-proof auctions~\cite{akbarpour2020credible, ferreira2020credible}. In a {\em credible auction}, the auctioneer is allowed to implement any ``safe deviation"---a deviation that cannot be detected by bidders---from the promised auction rule. Yet, in expectation, the auctioneer should achieve no more revenue than from implementing the promised auction. \citet{akbarpour2020credible} show that a family of ascending price auctions are the unique revenue-optimal, credible, and strategy-proof auctions for a computationally unbounded auctioneer. \citet{ferreira2020credible} propose the {\em Deferred Revelation Auction} (DRA), which is credible for an  auctioneer that is computationally bounded.\footnote{\citet{ferreira2020credible} also showed that DRA is credible only under standard auction theoretical assumptions on the distribution of bidder valuations. \citet{essaidi21credible} proposes the Ascending Deferred Revelation Auction (ADRA) which improves DRA by removing the distributional assumption. As a trade-off, ADRA requires a constant number of rounds (on expectation) to terminate rather than two rounds and bidders must deposit dynamic collateral---the collateral increases the longer a bidder participates in the mechanism---while DRA requires constant collateral.} However, it is not known how to implement the ascending-price auction or DRA on blockchains efficiently. Ideally, only information about winning transactions would be stored on the blockchain, whereas a naive implementation of the ascending-price auction or DRA would require information on all transactions. The challenge is that blockchain space is a limited and expensive resource. Each communication on the blockchain is stored in blocks that propagate to the whole network.

Another possible direction would be that of static posted-price mechanisms~\cite{chawla2010multi}, which are strategy-proof and credible. However, to extract good revenue and welfare, the posted price must be calibrated according to the distribution of user valuations. This is a challenge when user demand is highly volatile, making static posted-price mechanisms ill-suited to transaction fee markets.

In this paper, we study the design of {\em dynamic posted-price mechanisms}, where the price offered can adapt based only on information about winning bids from previous auctions. It is essential that the price adaptation only uses the value of winning bids and not all bids due to information asymmetry and the need for information efficiency. We seek to understand whether this ability to use adaptive prices enables approximately welfare-optimal credible mechanisms that provide a good user experience (incentive alignment).

Towards that direction, the recent {\em  EIP-1559 proposal} of \citet{buterin2019eip1559} changed the status quo, first-price auction on Ethereum to a hybrid between a first-price auction and a dynamic posted-price mechanism. In EIP-1559, each block has a fixed maximum capacity $m$ and a posted price $q>0$ that is a deterministic function of the previous block's utilization (described formally in Appendix \ref{sec:summary-mechanisms}). Every submitted transaction includes a {\em bid}, $b>0$, as well as a {\em bid cap}, $c>0$, which is the amount a user is willing to spend to have their transaction included. If a transaction's bid cap $c$ is greater than $b + q$, a miner may choose to include it in the block. A target utilization is set to $m/2$. The protocol increases the price in the next auction if block utilization is above the target and decreases the price if utilization is below the target. A winner pays $b + q$ to the blockchain, while miners receive only $b$ as revenue ($q$ is said to be ``burned" since no one receives it). Introducing a price floor via dynamic posted prices is likely to make payments more predictable, and this aspect of the design received good reception from users in a community outreach report~\cite{eip1559outreach}. Moreover, awarding only $b$ to miners makes EIP-1559 robust to collusion between miners and bidders~\cite{roughgarden2020transaction}.

When the posted price $q$ is above the market clearing price, EIP-1559's mechanism reduces to a dynamic posted-price mechanism similar to the mechanism we study: users will bid $b = 0$ and report a bid cap $c$ that equals their valuation for getting a transaction included in the next block.\footnote{Here, we assume there is no reserve price. EIP-1559 also supports a hard-coded reserve price in its first-price auction.}  A challenge, though, is that EIP-1559 still requires bidders to compete in a first-price auction when the posted price $q$ is below the market clearing price---the price where the demand meets the supply for block space.

\citet{roughgarden2020transaction} formalizes EIP-1559 as a decentralized auction game and shows that EIP-1559 is: incentive compatible for myopic bidders whenever $q$ is above the market clearing price; incentive compatible for myopic miners; and resistant to collusion between miners with bidders.\footnote{Myopic bidders and miners maximize their immediate utility from the next block.} \citet{roughgarden2020transaction} also proposes a {\em tipless mechanism} that modifies EIP-1559 so that the bid $b$ is hard-coded into the protocol. This change sacrifices collusion resistance in favor of the mechanism becoming incentive compatible for myopic bidders even when $q$ is below the market clearing price. Our work is complementary. We focus on designing dynamic posted-price mechanisms that provably converge to a price equilibrium. We propose alternative price update rules that improve upon EIP-1559 and converge to an equilibrium in instances where EIP-1559 is unstable. Our main observation is to design price update rules that use observable bids from prior blocks rather than rely simply on block utilization.  It is natural to ask if our mechanisms converge to a {\em good} price that approximates the market-clearing price. If the equilibrium price is too high, few bidders can afford to pay for block access, resulting in low welfare. Thus we also show that our mechanisms are approximately welfare optimal.

The main difference between the mechanisms studied here and the EIP-1559 proposal is how they handle the case where the posted price $q$ is below the market-clearing price (i.e., when the number of bidders willing to pay $q$  exceeds the supply for block space). While EIP-1559 forces users to compete in a first-price auction, we retain the simple dynamic posted-price incentives. Crucially, we need to ensure that even in the case of excess demand, the mechanism remains robust to manipulation by blockchain auctions, i.e., miners will choose to randomize the allocation to users who bid above the posted price.

Our analysis assumes that a miner is myopic about the decisions made regarding which transactions to include in a block whenever they are chosen to create a new block. This is a reasonable assumption for proof-of-work systems in which a new miner is, in effect, chosen at random in each time step to create a block. This creates uncertainty as to when a miner will next be ``chosen" and makes it reasonable that miners are myopic and seek to maximize immediate utility when given the opportunity.\footnote{This assumption is somewhat stronger in proof-of-stake (PoS) blockchains where miners only commit to a block at the moment that the block is revealed (instead of committing before the block is created). This significantly increases the ways in which miners can deviate from the intended protocol~\cite{ferreira2021proof}. However, miner myopia is still a reasonable assumption in PoS systems if we assume miners do not withhold blocks.}
\subsection{Brief technical overview}
We formalize the \emph{dynamic posted-price} (DPP) mechanism (Definition~\ref{def:dynamic-posted-price}) as a posted-price mechanism with a {\em  price update rule} and an {\em initial price} $q_0$. In each time step $t$, the mechanism announces a {\em price} $q_t$ and the miner whose turn it is can allocate block space to at most $m$ users from those with bids $b_i \geq q_t$ (henceforth, we refer to users as {\em bidders}). Given the winning bids at time step $t$, the update rule computes the subsequent price $q_{t+1}$. We seek a {\em good} posted price $z$,  which is a price that is above the market clearing price (the number of bidders willing to pay $z$ is at most $m$), but not too large so that the resulting welfare is approximately optimal (it can approximate the welfare of selling to the highest $m$ bidders).

\vspace{1mm}\noindent\textbf{How to pick the block size?} It is not immediately obvious why we should settle on auctioning blocks of constant size $m$. While allowing for larger blocks permits the blockchain to process more transactions per second, it also results in higher network propagation delay for new blocks. In the Ethereum blockchain, miners converge to an ideal block size via voting: each block is allowed to vote to increase $m$, decrease $m$, or keep $m$ unchanged for the subsequent block~\cite{o_neal_2020}. As a simplifying  assumption, we assume stakeholders agree on a block size $m$ that is time invariant.

\vspace{1mm}\noindent\textbf{Why use a price update rule?} Clearly, no single posted price can be good (in the sense defined above) if user demand changes over time. Given this, we set the goal as that of converging to a good posted price when demand is not changing. That is, the blockchain goes through epochs where the demand is stable and epochs where the demand is changing (e.g., during the launch of a new DeFi project). Although a posted price might not be good when demand is changing, once the demand stabilizes, we want the price to converge to a good price.

A dynamic posted-price mechanism is {\em asymptotically stable with respect to the distribution $F$ over bidder valuations and miners' randomness} if, starting from any initial price, the sequence of expected posted prices converges to an {\em equilibrium price} $z$; that is, if $z$ is the current posted price, then the expected value of next posted price is also $z$. An iterative algorithm that uses price update rules (like ours) allows us to naturally converge to a new equilibrium price whenever the distribution over bidder valuations changes. A mechanism that converges to a posted price $z$ is {\em approximately welfare optimal with respect to distribution $F$ over bidder valuations} if selling at this limit price $z$ gives a constant fraction of the optimal welfare (attained when selling to the top $m$ bidders).

A challenge is how to achieve incentive compatibility for bidders together with robustness against manipulation by miners, all at the same time as allowing for adaptive prices. One concern, for example, is how to determine the winning bids when there is excess demand, i.e., when there are more than $m$ bids with $b_i \geq q_t$. A naive approach 
would suggest that miners  use a {\em maximum value} (MV) allocation rule (Definition~\ref{def:mv}) and serve the highest bidders. However, this is not incentive-compatible for bidders. A bidder with value $v_i \geq q_t$ would have an incentive to bid $b_i = \infty$ to maximize their chances of receiving block space. To fix this issue, we propose the {\em random maximal} (RM) allocation rule (Definition~\ref{def:rm}): whenever more than $m$ bidders bid $b_i \geq q_t$, the miner must serve a uniformly random set with $m$ bidders among those with $b_i \geq q_t$.\footnote{This protocol requires only private randomization and not trusted, public randomness, which is an expensive resource on blockchains~\cite{ferreira2021proof}.} This removes incentives for bidders to overbid, since any bid above $q_t$ has an equal chance of being served. 

However, a remaining challenge is to design a price update rule that provides asymptotic stability while motivating a miner to follow the RM allocation. For example, the {\em Welfare-Based Dynamic Posted-Price} (WDPP) mechanism (Section~\ref{sec:wdpp}) converges exponentially quickly to a price equilibrium when the revenue curve, $\rev(q) = q \cdot \mde(q)$ (where $q$ is the price and $\mde(q)$ is the expected number of slots sold  at price $q$), is Lipschitz continuous. However, under the WDPP mechanism, miners strictly prefer to implement the MV allocation rule over the RM allocation rule (even if instructed otherwise). This drives the WDPP mechanism to not be incentive compatible for bidders---from the discussion above, bidders can benefit from overbidding.

We also study the \emph{Utilization-Based Dynamic Posted-Price} (UDPP) mechanism in Section~\ref{sec:udpp}, similar to the tipless mechanism~\cite{roughgarden2020transaction} and based on EIP-1559's price update rule. Still, UDPP differs from EIP-1559 in that it is purely a posted-price mechanism and not a hybrid between a first-price and posted-price mechanism (since it only requests bidders to report their budget). Moreover, it differs from the tipless mechanism by adopting a randomized allocation. Although UDPP is incentive compatible for bidders and protects the blockchain from deviations from revenue-maximizing miners, we show that posted prices from the UDPP mechanism can be unstable even when demand is smaller than supply (Proposition~\ref{prop:twdpp-stable}).

To combine the desirable properties of the WDPP and UDPP mechanism (and avoid their limitations), our main result is the design and analysis of the {\em Truncated Welfare-Based Dynamic Posted-Price} (TWDPP) mechanism, Section~\ref{sec:twdpp}, which is incentive compatible for myopic bidders, protects the blockchain from deviation from revenue-maximizing miners (Theorem~\ref{thm:twdpp-ic-dsic}), and, under natural assumptions, converges to an equilibrium price (Theorem~\ref{thm:twdpp-stability}) in which the mechanism provably obtains a worst-case $1/4$ of the optimal welfare (Theorem~\ref{thm:twdpp-welfare}). To prove stability, our main technical contribution is an iterative algorithm that, given oracle access to a Lipschitz continuous and strictly concave function $f$, converges to a fixed point of $f$ (Lemma~\ref{lemma:reduction}). Throughout the paper, we give intuitions for the theoretical results while deferring most proofs to the appendix.

In Section~\ref{sec:experimental-results}, we simulate the WDPP, UDPP, and TWDPP mechanisms. We find that the TWDPP mechanism achieves the strongest welfare of the group and obtains at least $4/5$ of the optimal welfare at the equilibrium price (outperforming the worst-case $1/4$ bound). 

\subsection{Related work}
We have surveyed most of the relevant work above, with~\citet{roughgarden2020transaction} being the closest to ours. \citet{leonardos2021dynamical} provides complementary stability conditions for the UDPP mechanism in the non-atomic model where the number of bidders is $n = m$ and $m \to \infty$. Our results do not make assumptions about $n$ and $m$, but require Lipschitz and concavity conditions on a quantity derived from the distribution of bidder valuations for any $n$ and $m$ (Theorem~\ref{thm:udpp-stability}) whereas their results do not.

We provide a detailed discussion of alternative transaction fee mechanism proposals in Appendix~\ref{sec:summary-mechanisms}. \citet{lavi2019redesigning} proposed the monopolistic price auction in a setting with unlimited block space where all transactions in a block pay the lowest transaction bid. \citet{yao2018incentive} answers several of the conjectures they introduce and shows that the monopolistic price auction is approximately incentive compatible under high demand ($n \to \infty$).

Google's adoption of the {\em generalized second-price auction} (GSP)~\cite{edelman2007internet} was a similar response to bid cycles that had plagued  first-price auctions~\cite{edelman2007strategic}. Unlike the blockchain setting, Google had the commitment power to implement a second-price auction. \citet{basu2019towards} proposed a modified GSP auction as a transaction fee mechanism, and show that it is incentive compatible for bidders but like \cite{lavi2019redesigning} requires the demand $n \to \infty$.

{\em Dynamic mechanism design}  has studied the design of mechanisms in  a setting where demand, preferences, and supply may change over time~\cite{parkes2004mdp, lavi2004competitive, hajiaghayi2005online, parkes2007online}. To the best of our knowledge, this earlier literature has not considered the problems of credibility and information asymmetry that are a defining feature concerning the design of transaction fee mechanisms.
Another body of  related work is that of {\em distributed algorithmic mechanism design} (DAMD), 
which was introduced   in the context of the Border Gateway Protocols (BGP) for Internet routing by
\citet{feigenbaum2004distributed}; this was
later extended to formalize the concept of the ``faithfulness" of distributed mechanisms, seeking to bring the intended rules into an {\em ex post} Nash equilibrium~\cite{shneidman2004specification,parkes2004distributed}.
In contrast to transaction fee mechanism design, this earlier work made critical use of redundant communication paths so that no one participant had a monopoly on information~\cite{monderer1999distributed}.

\section{Model}

We study a setting where there are $m$ identical {\em block slots} to allocate at each discrete time step. A set of bidders arrives in each step, and one miner is chosen randomly to become {\em active}. We assume miners are myopic and seek to maximize their immediate utility.

As we will explain, the miners can determine the allocation and can choose to deviate from the intended rules of a mechanism, but the miners cannot modify the payments conditioned on the allocation. These payments will be those of the dynamic, posted-price mechanism, and moreover, the prices will be updated across auctions according to the intended price update rule.

\subsection{Preliminaries}
The set of {\em active bidders} $M_t \subset \mathbb N_+$ denotes the bidders who arrive at time $t$. Bidder $i \in \cup_{t = 1}^\infty M_t$ has a {\em private value} $v_i \in \mathbb R_+$. We consider a stylized setting where bidders are myopic---bidders are only interested in receiving a slot at time $t$ and they do not participate at time $> t$.\footnote{\citet{huberman2017monopoly} considers the status quo first-price auction in the more general setting where bidders are interested in reducing their confirmation time rather than being included in the subsequent block.}

For each $i\in M_t$, let $b_i\in \mathbb{R_+}$ denote bidder $i$'s bid value. We assume lexicographic tie-breaking in the case of bidders with identical bid values. For a set of active bidders $M$, we define $M(q) := \{i \in M: b_i \geq q\}$. The \emph{allocation} $B_t \subseteq \mathbb N$ is a set of at most $m$ bids that receive a slot at time $t$.

We now define the allocation rule and payment rule for the family of dynamic mechanisms.
\begin{definition}[Allocation Rule]
An \emph{allocation rule} for a dynamic mechanism is a vector-valued function $\vec x$ from history $B_1, \ldots, B_{t-1}$ and active bidders $M_t$ to an indicator $x_i(B_1, \ldots, B_{t-1}, M_t) \in \{0, 1\}$ for each active bidder $i \in M_t$. The indicator is $1$ if and only if bidder $i \in M_t$ receives a slot at time $t$. An allocation rule is \emph{feasible} if $\sum_{i \in M_t} x_i(B_1, \ldots, B_{t-1}, M_t) \leq m$. An allocation rule is \emph{deterministic} if $x_i(B_1, \ldots, B_{t-1}, M_t)$ is a deterministic variable, and \emph{randomized} otherwise.
\end{definition}
\begin{definition}[Payment Rule]
A \emph{payment rule} for a dynamic mechanism is a vector-valued function $\vec p$ from history $B_1, \ldots, B_t$ to a payment $p_i(B_1, \ldots, B_t)$ for each bidder $i \in B_t$. 
\end{definition}
All bids are public to all miners, but are not available to the blockchain, i.e. the infrastructure determining payments and how prices are updated across auctions. Rather, the blockchain is only privy to information about the values of bids that are actually assigned a slot. Hence the payment rule can only depend on information disclosed in a block.

To distinguish the role of miners and the blockchain, we refer to mechanism $(\vec x, \vec p)$ as the {\em intended allocation rule} $\vec x$, to be implemented by miners, and the {\em intended payment rule} $\vec p$, to be implemented by the {\em pricing mechanism} (e.g., the blockchain). Although the pricing mechanism, which is provided by the blockchain, can faithfully commit to implement a payment rule, miners are unable to faithfully commit to implement the allocation rule. The pricing mechanism only learns the bids in $B_t$ and does not learn the bids in  $M_t \setminus B_t$. Thus, miners can try to change the allocation in their favor, including the effect this might have on future payments. As an example, a miner can impersonate a set of {\em fake bidders}, $F_t \subset \mathbb N_+$, at time $t$, and each fake bidder $i \in F_t$ can submit a \emph{fake bid} $b_i \in \mathbb R$ and be allocated a slot (see Example~\ref{example:1}). 

We assume that bidders have \emph{quasilinear utility}. That is, the utility of bidder $i \in M_t$ is
\begin{equation}\label{eq:utility}
u_i(v_i, b_i) := x_i(B_1, \ldots, B_{t-1}, M_t) \cdot (v_i - p_i(B_1, \ldots, B_t)).
\end{equation}
We  only consider mechanisms that transfer all payments received from  bidders to the active miner.\footnote{\citet{roughgarden2020transaction} shows that not awarding miners the revenue of the auction, which is what EIP-1559 does, is important to avoid collusion of bidders with miners. All of our results remain untouched if one modifies Equation~\ref{eq:miner-revenue} to transfer $\delta \in (0, 1]$ fraction of the revenue to the active miner.} The miner's utility is a function $u_0$ from allocation $B_t$ to the revenue of the miner (net any payments from fake bids):
\begin{equation}\label{eq:miner-revenue}
u_0(B_t) := \sum_{i \in B_t \setminus F_t} p_i(B_1, \ldots, B_t).
\end{equation}
\subsection{Decentralized multi-round auction game}

We now define the sequence of steps that models this transaction fee market.
\begin{definition}[Decentralized multi-round auction game]
The \emph{decentralized multi-round auction game} is a multi-round game between multiple miners and the pricing mechanism. We are given an allocation rule $\vec x$, a payment rule $\vec p$, and initially all miners are {\em inactive} and the time step is $t = 1$. The game proceeds as follows:
\begin{enumerate}[label=(\arabic*)]
\item \label{game-1} All miners observe $M_t$ and all bids $b_i \in M_t$.
\item  One miner is chosen at random to become \emph{active}.
\item \label{game-3} The active miner chooses and announces an allocation
$$B_t  := \{i \in M_t : x_i(B_1, \ldots, B_{t-1}, M_t) = 1\}.$$
\item The pricing mechanism observes $B_t$ and all bids $b_i \in B_t$.\footnote{The integrity of bid values can be enforced via digital signatures.}
\item \label{game-5} For each bidder $i \in B_t$, the pricing mechanism charges $p_i(B_1,\ldots, B_t)$.
\item \label{game-6} The active miner receives a payment $\sum_{i \in B_t} p_i(B_1, \ldots, B_t)$.
\item The active miner becomes inactive. Increment $t$ and return to step~\ref{game-1}.
\end{enumerate}
\end{definition}
Unlike standard mechanism design~\cite{myerson1981optimal}, miners are responsible for choosing $B_t$, and yet at the same time, cannot commit to implementing a particular allocation rule $\vec x$; that is, an active miner can corrupt step \ref{game-3} of the game in any way they want. Miners may even allocate slots to fake bidders, $i \in F_t$. 

Because the pricing mechanism can only observe bids in $B_t$ and cannot observe $M_t \setminus B_t$, the payments must depend  only on $B_1, \ldots, B_t$. 

Blockchains can commit to executing honest code and can faithfully implement steps \ref{game-5} and \ref{game-6}. Thus, miners can manipulate the allocation $B_t$ of the intended mechanism, but not what the pricing mechanism computes once it learns $B_t$.
\begin{example}\label{example:1}
Suppose there is a single slot for sale, and the mechanism $(\vec x, \vec p)$ is a first-price auction (the highest bidder receives the slot and pays his bid). Suppose there are two active bidders at time $t$, $M_t = \{1, 2\}$, with identical values $v_1 = v_2 = 2$ and distinct bids $b_1 = 1$, $b_2 = 2$. While the values are private to bidders, the bids become known to miners at Step~\ref{game-1}. At Step~\ref{game-3}, the active miner is instructed by the allocation rule to allocate the slot to bidder 2, but can make other choices. For example, the active miner can allocate the slot to bidder 1 and receive utility $\$1$; allocate to bidder 2 and receive utility $\$2$; allocate the slot to neither bidder and receive utility $\$0$; or even allocate the slot to a fake bidder with some bid $b$ and receive a utility of $v - v = \$0$ (i.e., they pay the mechanism at Step~\ref{game-5}, but the mechanism returns the payment at Step~\ref{game-6}). In this example, the active miner's optimal action is to allocate the slot to bidder 2 and receive utility $\$2$. Thus, in this example, it is in the best interest for active miners to implement the intended mechanism.
\end{example}

\subsection{Mechanism design}
We study the stability and the social welfare of a mechanism in the {\em probabilistic setting}. For simplicity, assume $n$ bidder values are drawn i.i.d. from a continuous {\em cumulative density function} $F$:
\begin{equation}\label{eq:cdf}
F(x) := Pr[v \leq x].
\end{equation}
Although our welfare guarantees require the i.i.d. assumption, our conditions for stability do not. Instead, our conditions for stability will require Lipschitz and concavity conditions on quantities derived from the distribution where bidder valuations are drawn. We define $\vec v := \{v_1, v_2, \ldots, v_n\} \in \mathbb R_+^n$ as an $n$-dimensional random variable, where $v_1, v_2, \ldots, v_n$ are drawn i.i.d. from $F$.
\begin{definition}[Social Welfare]\label{def:welfare}
For a mechanism $(\vec x, \vec p)$ with history $B_1, B_2, \ldots, B_{t-1}$, let $M_t = [n] = \{1, 2, \ldots, n\}$ be a set of bidders with values $\vec v = \{v_1, v_2, \ldots, v_n\}$ drawn i.i.d. from $F$. The expected \emph{social welfare} at time $t$ with respect to allocation rule $\vec x$ is
\begin{equation}\label{eq:social-welfare}
\welfare_t(\vec x) := \mathbb E_{\vec v \sim F^n}\left[\sum_{i = 1}^n v_i \cdot x_i(B_1, \ldots, B_{t-1}, M_t) \bigg| B_1, \ldots, B_{t-1}\right].
\end{equation}
We drop the subscript $t$ and write $\welfare(\vec x)$ when the time step is clear from context. A mechanism is {\em efficient} if it extracts the maximum possible social welfare. Let $\vec x^*(\vec v)$ be the welfare-maximizing feasible allocation for a particular value profile $\vec v$, with $\sum_{i = 1}^n x_i^*(\vec v) = m$. The \emph{optimal expected social welfare} is:
\begin{equation}\label{eq:opt-welfare}
\opt := \e{\sum_{i = 1}^n v_i \cdot x_i^*(\vec v)}.
\end{equation}
\end{definition}
We will consider a mechanism as having a \emph{good user experience} if it is {\em ex post Nash incentive compatible (IC)} and {\em individually rational (IR)} for bidders---bidders have the incentive to bid their valuation and pay at most their bid. 
\begin{definition}[Ex Post Nash Equilibrium]\label{def:epne}
A \emph{strategy} for bidder $i$ is a function $b_i^*$ from value $v_i$ to bid $b_i^*(v_i)$. A strategy profile $b^*(\cdot)$ is an \emph{ex post Nash equilibrium} (EPNE) for a mechanism $(\vec x, \vec p)$ if for all histories $B_1, \ldots, B_{t-1}$, for all active bidders $M_t$, for all bidders $i \in M_t$, for all values $v_i$, bidding $b_i^*(v_i)$  maximizes the utility~\eqref{eq:utility} of bidder $i$ conditioned on all $j \neq i \in M_t$ following strategy $b_j^*(v_j)$.
\end{definition}
\begin{definition}[Ex Post Incentive Compatible]\label{def:strategyproof}
A mechanism $(\vec x, \vec p)$ is \emph{ex post incentive compatible} (IC) for bidders if it has a \emph{truthful ex post} Nash equilibrium where all real bidders bid their true value. That is, $b_i^*(v_i) = v_i$ for all $i \in \cup_{t = 1}^\infty M_t$.
\end{definition}
Note dominant-strategy IC rather than {\em ex post} IC is too strong for the blockchain setting because the pool of bids is public and, in principle, bidders could choose to condition their strategies on the reports of others. Yet {\em ex post} IC is an appropriate and typical solution concept for dynamic settings~\cite{cavallo10,bergemann10}, and provides a useful strategic simplification and improvement in ease of use: it is optimal for a bidder to report truthfully, no matter its actual value and the values of others, as long as others are also truthful.
\begin{definition}[Individual Rationality (IR)]
A mechanism $(\vec x, \vec p)$ is \emph{individually rational} if for all histories $B_1, B_2, \ldots$, bidders pay at most their bids: $p_i(B_1, B_2, \ldots) \leq b_i$ for all bidders $i$.
\end{definition}
\begin{definition}[Dominant Strategy Incentive Compatible (DSIC) for Myopic Miners]\label{def:credible}
A mechanism $(\vec x, \vec p)$ is \emph{DSIC for myopic miners} if for all histories $B_1, \ldots, B_{t-1}$, for all active bidders $M_t$, the active miner at time $t$ maximizes their utility~\eqref{eq:miner-revenue} by implementing the allocation rule $\vec x$. That is, the miner sets $F = \emptyset$ and chooses $B_t = \{i \in M_t : x_i(B_1, \ldots, B_{t-1}, M_t) = 1\}$.
\end{definition}
To ensure the bidder's myopia is an adequate assumption, we will aim to design dynamic mechanisms that are stable. To be concrete: whenever the demand for blocks is at a stable state, we will require the mechanism to reach a stable state, at which point a bidder will receive the same expected utility regardless of the time they submit a bid. For simplicity, we first give a notion of stability suitable for the deterministic setting and in Section~\ref{sec:dpp}, we provide a more general definition for the probabilistic setting.
\begin{definition}[Deterministic Stability]\label{def:deterministic-stability}
Assume the same set $M$ of bidders with bid profile $\vec b$ arrives at all time steps, i.e. $M_t = M$ for all $t$. The mechanism $(\vec x, \vec p)$ is {\em stable} with respect to the deterministic demand $M$ if for all bidders $i \in M$, the allocation rule $x_i(B_1, B_2, \ldots, B_{t-1}, M_t)$ converges to $x_i^*(\vec b)$ and the payment rule $p_i(B_1, B_2, \ldots, B_t)$ converges to $p_i^*(\vec b)$ as $t \to \infty$.
\end{definition}
Note if mechanism $(\vec x, \vec p)$ is stable, then the dynamic mechanism converges to a static mechanism $(\vec x^*, \vec p^*)$. If $\vec b$ is the bid profile of all bidders, then the utility of bidder $i$ with value $v_i$ is $(v_i - p_i^*(\vec b)) \cdot x_i^*(\vec b)$ which is independent of the time they choose to participate which motivates bidder's myopia.

Below we summarize the desiderata for the design of transaction fee mechanisms for decentralized blockchains:
\begin{itemize}\itemsep0em
    \item \textbf{Communication Complexity.} The mechanism requires $O(m)$ communication between bidders and the blockchain.\footnote{The first-price, second-price, and posted-price mechanisms have communication complexity $O(m)$. As a distinction, the communication cost between miners and the blockchain on known implementations of credible, strategyproof, optimal auctions would scale with $n$, the number of bidders, not the number of slots.}
    \item \textbf{Ex Post Incentive Compatible for Bidders.} The mechanism has a truthful {\em ex post} Nash equilibrium for  bidders.
    \item \textbf{Individually Rational.} If bidders receive a slot, they pay at most their bid; otherwise, they pay nothing.
    \item \textbf{DSIC for Myopic Miners.} In each period, it is a weakly dominant strategy for miners to implement the intended allocation rule.
    
    \item \textbf{Stable.} The mechanism converges to a stable state.
    \item \textbf{Approximately Welfare Optimal.} If in each period, $n$ i.i.d. bidders are drawn from distribution $F$, then the mechanism converges to an equilibrium price.\footnote{We consider asymptotic stability with respect to the distribution $F$ over bidder valuations as our formal definition for a stable equilibrium (Definition~\ref{def:stability}).} At the equilibrium price, the expected welfare is $\Omega(\opt)$, where $\mathit{\opt}$ is the optimal expected social welfare.
\end{itemize}

\subsection{Illustrative mechanisms}
Next, to motivate the design of dynamic mechanisms, we highlight undesirable properties from traditional static mechanisms.  Table~\ref{tab:single-block-mechanisms} summarizes the properties of existing proposals. We also provide further discussion in Appendix~\ref{sec:summary-mechanisms}. 
\begin{table}
\begin{tabularx}{\linewidth}{XXX}
\toprule
\textbf{Static mechanisms} & \textbf{IC for bidders} & \textbf{DSIC for miners} \\
\midrule
\midrule
First-Price~\eqref{def:first-price} & No & Yes \\
Second-Price~\eqref{def:second-price} & Yes & No\\
DPP~\eqref{def:posted-price} & No & Yes \\
RPP~\eqref{def:posted-price} & Yes & Yes \\
EIP-1559~\eqref{sec:eip-1559} & Sometimes & Yes \\
Monopolistic Price~\eqref{sec:monopolistic}  & Approximately & Yes \\
RSOP~\eqref{sec:rsop} & Yes & No\\
Modified GSP~\eqref{sec:gsp} & Approximately & Approximately\\
\bottomrule
\end{tabularx}
\caption{Bidder payments and miner revenue under various static mechanisms. $q_t$ denotes a posted-price. At time $t$, $M_t$ denote the set of active bidders and $B_t$ denotes the allocation.}
\label{tab:single-block-mechanisms}
\end{table}
\begin{definition}[Static mechanism; Dynamic mechanism]\label{def:static}
A mechanism $(\vec x, \vec p)$ is \emph{static} if its allocation and payment rule depend only on the most recent bids and allocation. That is, for all histories $B_1, \ldots, B_{t-1}$, we have $x_i(B_1, \ldots, B_{t-1}, M_t) = x_i(M_t)$ and $p_i(B_1, \ldots, B_t) = p_i(B_t)$. A mechanism is \emph{dynamic} if it is not static.
\end{definition}
\begin{definition}[First-price auction]\label{def:first-price}
A \emph{first-price} (FP) auction is a static mechanism where the top $m$ bidders receive a slot and pay their bid.
\end{definition}

\citet{nakamoto2008bitcoin} proposed the first-price auction as the default transaction fee mechanism for Bitcoin. 

An important property of the first-price auction is that it is DSIC for myopic miners. Miners can only lose revenue by replacing a real bid $i \in M_t$ with a fake bid $j \in F_t$. As a downside, the first-price auction is known to provide a bad user experience as a transaction fee mechanism because it has no {\em ex post} Nash equilibrium. A bidder wants to bid the minimum possible that would allow her to win a slot, but that requires the bidder to monitor all other bids and adapt her strategy.
\begin{definition}[Second-price auction]\label{def:second-price}
A \emph{second-price} (SP) auction is a static mechanism where the top $m$ bidders receive a slot and pay the $(m+1)$-st highest bid.
\end{definition}
Unlike the first-price auction, truthful bidding $b_i(v_i) = v_i$ is an {\em ex post} Nash equilibrium for the second-price auction, making it easy for  bidders to participate. However, to implement its payment rule, miners must commit to truthfully reporting the value of  $(m+1)$-st highest bid with the mechanism (note the payment rule depends not only on the allocation $B_t$ but also on unallocated transactions $M_t \setminus B_t$). Thus miners are better off by reporting a fake bid that equals the $m$-th highest bid, and the second-price auction is not DSIC for myopic miners.

Next, we consider posted-price mechanisms and two possible allocation rules.
\begin{definition}[Maximal Allocation]
For a set of active bidders $M$, we say $S \subseteq M$ is a \emph{maximal allocation} with respect to $M$ if the following is true:
\begin{itemize}\itemsep0em
\item $S$ contains at most $m$ bidders.
\item If $S$ contains less than $m$ bidders, $S = M$.
\end{itemize} 
\end{definition}

Let $2^M$ denote the {\em  collection of maximal allocations} with respect to $M$:
\begin{equation}\label{eq:maximal-collection}
2^{M} := \{B \subseteq M : \text{$|B| \leq m$ is a maximal allocation}\}.
\end{equation}
\begin{definition}[Maximum value allocation rule]\label{def:mv}
The \emph{maximum value} (MV) allocation rule $\vec x^{MV}(M)$ selects a maximal allocation $B \in 2^M$ that maximizes bid values: $B = \arg\max_{B' \in 2^M} \sum_{i \in B'} b_i$.
\end{definition}
Let $\mathit{Unif}(\{B_1, \ldots, B_\ell\})$ denotes the uniformly random distribution over $\{B_1, \ldots, B_\ell\}$.
\begin{definition}[Random maximal allocation  rule]\label{def:rm}
The \emph{random maximal} (RM) allocation rule $\vec x^{RM}(M)$ samples a maximal allocation $B \in 2^M$ from the distribution $\unif(2^M)$.
\end{definition}
\begin{definition}[Posted-price mechanism]\label{def:posted-price} A \emph{posted-price} (PP) mechanism with posted price $q$ is a static mechanism $(\vec x, \vec p)$ where each bidder that receives a slot pays $q$ and all other bidders pay nothing. We refer to $(\vec x(q), \vec p(q))$ as the posted-price mechanism with posted price $q$. A posted-price mechanism with posted price $q$ is a \emph{randomized posted-price} (RPP) mechanism if for all $M$, $\vec x(M) = \vec x^{RM}(M(q))$ is the RM allocation rule on bids $M(q) = \{i \in M : b_i \geq q\}$. A posted-price mechanism with posted price $q$ is a \emph{deterministic posted-price} (DPP) mechanism if for all $M$, $\vec x(M) = \vec x^{MV}(M(q))$ is the MV allocation rule on bids $M(q)$.
\end{definition}
Posted-price mechanisms are DSIC for myopic miners since they commit to a fixed-price that miners cannot influence, and thus, the best a miner can do is accept a maximal allocation of bids priced at or above the price $q$. 

Depending on how the posted-price mechanism allocates slots, the mechanism can also be incentive compatible for bidders. The posted-price mechanism with the RM allocation rule is {\em ex post} IC since any bidder with $b_i \geq q$ has an equal chance to receive a slot irrespective of its bid. On the other hand, the posted-price mechanism with the MV allocation rule is not {\em ex post} IC. Whenever more than $m$ bidders have value $v_i \geq q$, there exists a bidder with value $v_i \geq q$ that will not receive a slot. If that bidder bids $b_i = \infty$ instead of $v_i$, they receive the slot and pay $q \leq v_i$.
\subsection{Discussion}
The first-price, second-price, and posted-price mechanisms are individually rational and require little communication with the blockchain. This communication efficiency is a first-order concern for decentralized blockchains because any information stored in the blockchain must propagate through the whole network. In particular, these mechanisms require miners to report $O(m)$ bids to the mechanism to implement the payment rule---the communication scales as the number of slots, not the number of bidders.  

Posted-price mechanisms are also naturally robust to miners' strategic behavior and are incentive compatible for bidders provided the posted price is high enough (so that the supply is greater than the demand at the posted price). While first-price auctions are not incentive compatible for bidders, they are resilient to manipulation by miners. Finally, second-price auctions are incentive compatible for bidders but not resilient to manipulation by miners.

In terms of social welfare, the second-price auction is efficient since truthful bidding is a dominant strategy and the highest bidders receive block slots. In theory, first-price auctions are also efficient assuming bids are drawn i.i.d. from some distribution $F$ and under the assumption bidders follow a Bayes-Nash equilibrium. However, bids in a first-price auction do not typically converge to an equilibrium in repeated settings, as was observed for paid search engines~\cite{edelman2007strategic}, and the current blockchain transaction fee market is an example of the kind of instability that can occur. Sequential posted-price mechanisms---where the auctioneer visits bidders to make a take-or-leave offer---can obtain $1/2$ of the optimal welfare, but require prior knowledge on the distribution $F$ and the demand $n$~\cite{chawla2010multi, kleinberg2012matroid}.

\section{Dynamic Posted-Price Mechanisms}\label{sec:dpp}

The rationale behind a dynamic posted-price mechanism is that even though the demand curve is unknown and dynamic over time, we can hope to learn a proxy for the market-clearing price from transactions included in preceding blocks. We must be careful, though, since we must align the mechanism with miners' incentives to implement the allocation rule faithfully.
\begin{definition}[Price Update Rule]\label{def:update-rule}
A \emph{price update rule} is a real-valued function $T$ from positive $q$ and feasible allocation $B$ with respect to $q$ to the subsequent price $T(q, B)$ where $b_i \geq q$ for all $i \in B$.
\end{definition}
\vspace{1mm}\noindent\textit{Remark:} we use $B$ to refer to the bidders that receive a slot. Given $i \in B$, the update rule has access to the bid $b_i$ of bidder $i$.
\begin{definition}[Dynamic Posted-Price Mechanism]\label{def:dynamic-posted-price}
A \emph{dynamic posted-price} (DPP) mechanism $(\vec x, \vec p, T)$ is a dynamic mechanism endowed with a price update rule $T$. Let positive $q_0$ denote the initial posted price. Given history $B_1, B_2, \ldots, B_{t-1}$ and active bidders $M_t$, define the posted price $q_t := T(q_{t-1}, B_{t-1})$ at time $t$. The mechanism $(\vec x, \vec p, T)$ is a {\em randomized dynamic posted-price} (RDPP) mechanism if, at time step $t$, it implements the RM allocation rule:
$$\vec x(B_1, B_2, \ldots, B_{t-1}, M_t) := \vec x^{RM}(M_t(q_t)).$$
The mechanism $(\vec x, \vec p, T)$ is a {\em deterministic dynamic posted-price} (DDPP) mechanism if, at time step $t$, it implements the MV allocation rule:
$$\vec x(B_1, B_2, \ldots, B_{t-1}, M_t) := \vec x^{MV}(M_t(q_t)).$$
\end{definition}
\smallskip

We study the {\em stability} of the expected value of the random operator $T(q, B)$ when at each time step, $n$ bidders with values $\vec v = (v_1, v_2, \ldots, v_n)$ are drawn from distribution $F^n$, with the assumption that bidders are truthful (bid their value). We define the real-valued function $E_T$ to be
\begin{equation}\label{eq:et}
E_T(q) := \mathbb E_{\vec v \sim F^n}\left[T(q, B)\right],
\end{equation}
where the expectation taken is over $v_1, v_2, \ldots, v_n$ and over the allocation
$$B = \{i \in M: x_i(q) = 1\}.$$
Before we formally define stability for a DPP mechanism, we define the fixed point iteration of a real-valued function. We use $(X, d)$ to refer to a complete metric space where $X \subseteq \mathbb R$ and $d$ is the Euclidean norm.
\begin{definition}[$n$-iterate]
The \emph{$n$-iterate} of function $f : X \to X$, for integer $n \geq 0$, is the function 
\begin{equation}
f^n(q) := \begin{cases} 
 q \quad & \text{for $n = 0$},\\
 f(f^{n-1}(q)) \quad & \text{for $n \geq 1$}. 
\end{cases}
\end{equation}
\end{definition}
A point $x$ is a \emph{fixed point} for function $f : X \to X$ if and only if $f(x) = x$.
\begin{definition}[Fixed-point iteration]\label{def:fixed point-iteration}
Given a continuous function $f : X \to X$ and an initial point $q_0 \in X$, the \emph{fixed-point iteration} of $f$ starting from $q_0$ is the sequence of function evaluations
\begin{equation}
q_0, f(q_0), f^2(q_0), f^3(q_0), \ldots
\end{equation}
If the sequence converges to some $q^*$, then from the continuity of $f$, $q^*$ is a fixed point of $f$.
\end{definition}
\begin{definition}[Asymptotic Stability]\label{def:stability}
A DPP mechanism $(\vec x, \vec p, T)$ is {\em asymptotically stable} with respect to distribution $F$ over $n$ bidder valuations if, for {\em any} positive initial price $q_0$, the fixed-point iteration of $E_T$ starting from $q_0$ converges to the {\em equilibrium price} $q^*$. A DDP mechanism is {\em unstable} if it is not asymptotically stable.
\end{definition}
\begin{definition}[Approximately Welfare Optimal at Equilibrium]\label{def:welfare-equilibrium}
A DPP mechanism is {\em approximately welfare optimal} if it is asymptotically stable with equilibrium price $q^*$ and the welfare at equilibrium $\welfare(\vec x(q^*)) \geq \Omega(\opt)$ is a constant approximation of the optimal welfare.
\end{definition}
There are two main reasons why a DPP mechanism could be unstable: prices could grow unbounded, or prices could oscillate in a closed orbit. The proofs of Propositions~\ref{prop:instability} and~\ref{prop:twdpp-stable} give examples of mechanisms that are unstable when the distribution $F$ is a point mass.

The convergence of prices is also important in justifying our assumption that bidders are myopic and will immediately submit bids for the next auction. If prices were to oscillate, bidders might seek to time the moment they enter their bids, waiting for moments where prices are lower, driving more instability and uncertainty in the market.

\section{Asymptotically Stable Update Rules via Fixed Point Theory}\label{sec:fixed-point}
There is a rich theory on the study of fixed points ~\cite{collatz2014functional, herzog2013fixed} for monotone decreasing operators and in many settings, the existence of a fixed point is guaranteed. In general, we do not expect a function $f : \mathbb R \to \mathbb R$ to have a fixed point, for example $f(x) = x + 1$. However, if a continuous function $f$ is non-increasing, then $f$ always intersects the line $y(x) = x$. Let $z$ denote the $x$ value corresponding to this intersection. Then $z$ is a fixed point for function $f$, since $f(z) = y(z) = z$.
\begin{definition}[Lipschitz continuity]\label{def:lipschitz}
A real-valued function $f : X \to \real$ is {\em $L$-Lipschitz} if for all $x, y \in X$, $|f(x) - f(y)| \leq L\cdot |x- y|$. We say $f$ is {\em Lipschitz continuous} if there is a positive constant $L$ such that $f$ is $L$-Lipschitz. We say that $f$ is a {\em constant function} if and only if $f$ is $0$-Lipschitz.
\end{definition}
\begin{definition}[Contractive mapping]\label{def:contractive-mapping}
A {\em contractive mapping} is a function $f : X \to X$ with the property that $f$ is Lipschitz continuous with constant $0 \leq L < 1$.
\end{definition}
\begin{definition}[Monotone mixture]\label{def:monotone-mixture}
A {\em monotone mixture} $f$ with {\em kernel} $g : X \to X$ is a function $f(x) = \alpha \cdot g(x) + (1-\alpha)\cdot x$ where the kernel is nonconstant, nonincreasing and $\alpha \in [0, 1]$ is the convergence parameter.
\end{definition}
\begin{observation}\label{obs:equal-fixed point}
If $f(x) = \alpha g(x) + (1-\alpha) x$, then $x$ is a fixed point for $f$ if and only if $x$ is a fixed point for $g$.
\end{observation}
\begin{proof}
Assume $x = f(x)$, then $g(x) = x$. Conversely, assume $x = g(x)$, then $f(x) = x$.
\end{proof}
Thus, for any continuous nonincreasing function $g$, we can use Observation~\ref{obs:equal-fixed point} to construct a function $f$ that has the same fixed point as $g$. Then, we can search for a fixed point for $g$ by searching for a fixed point for $f$. For this, we use the following well known fact:
\begin{lemma}\label{lemma:contractive-update-rule}
If $f(x) = \alpha g(x) + (1-\alpha)x$ is a monotone mixture with a nonconstant $L$-Lipschitz kernel $g$, then for all $0 \leq \alpha \leq \frac{1}{L + 1}$, $f$ is $(1-\alpha)$-Lipschitz. Additionally, if $f$ maps $X$ to itself, the restriction of $f$ to $X$ is a contractive mapping.
\end{lemma}
\begin{proof}
Fix $x, y \in X$ and w.l.o.g. assume $x < y$, then $g(x) \geq g(y)$ is nonincreasing and
\begin{align*}
|f(x) - f(y)| &= |\alpha (g(x) - g(y)) + (1-\alpha)(x-y)|\\
&= |\alpha (g(x) - g(y)) - (1-\alpha)(y - x)|\\
&\leq \max\{\alpha |g(x) - g(y)|, (1-\alpha)|x- y|\}\\
&\leq \max\{\alpha L |x-y|, (1-\alpha)|x-y|\}\\
&\leq (1-\alpha)|x-y|.
\end{align*}
The last step is from the assumption $\alpha \cdot L + \alpha \leq 1$. This proves $f$ is $(1-\alpha)$-Lipschitz. For the ``additionally'' part, if $f$ maps $X$ to itself, the fact that $g$ is non-constant implies $L > 0$, thus $\alpha < 1$ and the restriction of $f$ to $X$ is a contraction.
\end{proof}
\begin{lemma}[Banach Contraction Principle~\cite{royden1988real}]\label{lemma:banach-contraction}
If a real-valued function $f : X \to X$ is a contraction, then $f$ has a unique fixed point $x^*$. To find $x^*$, start at any $x_0 \in X$, then $\lim_{n\to\infty} f^n(x_0) = x^*$.
\end{lemma}
Thus, if we design a DPP mechanism $(\vec x, \vec p, T)$ such that for distribution $F$ over bidder valuations the expected value $E_T$ is a monotone mixture with an $L$-Lipschitz kernel and convergence parameter $\alpha \leq \frac{1}{L+1}$, we will derive that $(\vec x, \vec p, T)$ is asymptotically stable with respect to $F$. This will be the main idea behind the proof that the WDPP mechanism is stable (Theorem~\ref{thm:wdpp}). The proof of stability of the UDPP mechanism and the TWDPP mechanism will be more involved because their kernels will not be monotone mixtures, but the proof will still rely on the Banach contraction principle.
\section{Welfare-Based Dynamic Posted-Price Mechanism}\label{sec:wdpp}
This section describes our first proposal and gives sufficient conditions for our mechanism to be asymptotically stable and to obtain a $1/2$-approximation of the optimal welfare. 

To maximize social welfare, we must guarantee the highest-valued transactions are included in each block. The challenge is that whenever there are more than $m$ bidders with bids $b_i \geq q_{t-1}$, the mechanism is not {\em ex post} IC for bidders when miners use the MV allocation rule. If we use the RM allocation rule, a low bid may receive a slot in place of a high bid. In this case, the mechanism must increase the posted price. A welfare-based approach uses bids of allocated bidders as a proxy for the degree of excess demand:
\begin{definition}[Welfare-based update rule]\label{def:welfare-rule} The {\em $\alpha$-welfare-based update rule} is the update rule:
\begin{equation}
\tw(q, B) := \alpha \frac{1}{m}\sum_{i \in B} b_i + (1-\alpha) q,
\end{equation}
where $\alpha \in [0, 1]$ is the convergence parameter.
\end{definition}
Observe that whenever $|B_t| = m$, $q_t = \tw(q_{t-1}, B_{t-1}) > q_{t-1}$ unless $b_i = q_t$ for all $i \in B_{t-1}$. Thus the hope is that the mechanism converges to a price $z$ at least as big as the market clearing price---the price where at most $m$ bidders have value $v_i \geq z$.

Unfortunately, miners' myopia is an unreasonable assumption for a dynamic posted-price mechanism with the RM allocation rule endowed with the welfare-based update rule. To see this, note that if the current price $q_{t-1}$ is smaller than the market-clearing price, there exists a price $q > q_{t-1}$ that also sells to $m$ buyers. Thus under a demand surge, the active miner might have a strict preference for implementing the MV allocation rule instead of the RM allocation rule to drive prices up for the next time when they become the active miner. Once miners signal a preference for adding the highest bids, each bidder with $v_i \geq q_t$ would have a dominant strategy to bid $b_i > v_i$ to maximize the probability of being selected. Thus, the dynamic posted-price mechanism would not be {\em ex post} IC for myopic bidders.\footnote{A possible counterargument to bidders overbidding comes from {\em Mempools}---global decentralized databases of pending transactions---where Bitcoin and other blockchains store bids until they are processed and added to the blockchain or removed after 48 hours. With Mempools, there is {\em carryover} from $M_t \setminus B_t$ to $M_{t+1}$ (if not served, bidders participate in the next step). Thus if a bidder bids $b_i > v_i$ in the hope of increasing their chances of being included in $B_t$, there is a risk their bid will carry over to the next time step. Since there is uncertainty about future prices (since $M_{t+1}, M_{t+2}, \ldots$ are unknown at time $t$), there is a risk bidder $i$ could pay $b_i > v_i$ and receive a negative utility. It is hard to model carryover, and incorporating Mempools into the model is an interesting avenue for future work. Here, we assume each bidder only participates once.}

Since a DPP mechanism with the welfare-based update rule and the RM allocation rule reduces to a DPP mechanism with the MV allocation rule, we study a DPP mechanism endowed with the welfare-based update rule and the MV allocation rule:
\begin{definition}[Welfare-Based Dynamic Posted-Price mechanism]\label{def:wdpp}
The {\em welfare-based dynamic posted-price} (WDPP) mechanism is a DDPP mechanism $(\vec x, \vec p, \tw)$ endowed with the welfare-based update rule $\tw$.
\end{definition}
The WDPP mechanism is not {\em ex post IC}. However, we show the WDPP mechanism is asymptotically stable assuming truthful bidding as a warm-up to studying {\em ex post IC} mechanisms in the following sections.
\subsection{Asymptotic stability of the WDPP mechanism}
We will now describe our main results for the WDPP mechanism. In particular, we show that by setting a sufficiently small $\alpha$, the WDPP mechanism is asymptotically stable and approximately welfare optimal under the mild assumption the revenue curve is Lipschitz continuous.
\begin{definition}[Demand Curve]
The {\em demand at price $q$} is the random variable $N(q) := \sum_{i = 1}^n \ind{v_i \geq q}$. A {\em demand curve} is a non-increasing real-valued function $D$ mapping a price $q$ to the expected number of buyers $D(q) := \e{N(q)}$ that would purchase at price $q$. The {\em demand curve with limited supply} is $\mde(p) :=  \e{\min\{m, N(q)\}}$.
\end{definition}
\begin{definition}[Revenue Curve]
A {\em revenue curve} is a real-valued function from price $q$ to the expected revenue $\rev(q) := q \cdot \mde(q)$ when selling a supply of $m$ slots at price $q$.
\end{definition}
We can show the WDPP mechanism is asymptotically stable as long as the revenue curve $\rev(q)$ is Lipschitz continuous and $\alpha$ is upper bounded by a function of the Lipschitz constant of $\rev(q)$.
\begin{theorem}\label{thm:wdpp}
If the revenue curve $\rev$ is $L$-Lipschitz and $\alpha \leq \frac{1}{L/m + 1}$, the WDPP mechanism $\left(\vec x, \vec p, \tw\right)$ is asymptotically stable with respect to $F^n$ and has a unique equilibrium price $q^*$. Moreover, it is approximately welfare optimal: $\welfare(\vec x(q^*)) \geq \frac{\opt}{2}$.
\end{theorem}
We defer the proof to the appendix. As a proof sketch recall that a mechanism is asymptotically stable if the iteration $E_{\tw}(q_0), E_{\tw}(q_1), \ldots$ converges to an equilibrium price $z$. Expanding $E_{\tw}(q_0)$, we have
$$E_{\tw}(q_0) = \alpha \frac{\welfare(\vec x(q_0))}{m} + (1-\alpha)q_0.$$
Since miners allocate block space to the highest bidders, we have that $\welfare(\vec x(q_0))$ is a decreasing function of $q_0$. Therefore, $E_{\tw}$ is a monotone mixture as defined in Section~\ref{sec:fixed-point}. One can show that if $\rev(q_0)$ is Lipschitz continuous, then so is the $\welfare(\vec x(q_0))$. Thus we can use the ideas of Section~\ref{sec:fixed-point} to prove that the WDPP mechanism is asymptotically stable.

To prove that the WDPP mechanism is approximately welfare optimal, recall that we need to show $\welfare(\vec x(z)) \geq \Omega(OPT)$ where $z$ is an equilibrium price. Note that $z$ is an equilibrium price if and only if  $E_{\tw}(z) = z$ and observe that $E{\tw}(z) = \alpha f(z) + (1-\alpha) z$ is a convex combination of $z$ with a function $f$ of $z$. Observe that $z$ is also a fixed point of $f$: $f(z) = z$. For the WDPP mechanism, $z = f(z) = \frac{\welfare(\vec x(z))}{m}$, which we use when showing $\welfare(\vec x(z)) \geq \frac{OPT}{2}$.

\vspace{1mm}\noindent\textit{Remark.} For all dynamic posted-price mechanisms $(\vec x, \vec p, T)$ we consider, we will always have that $E_{T}(z) = \alpha f(z) + (1-\alpha) z$ (i.e. $E_T$ is a convex combination of $z$ with some function $f$ of $z$). Thus if the price update rule $T$ is asymptotically stable, it converges to an equilibrium price $z$ where $f(z) = z$, which we use to show the mechanism is approximately welfare optimal.
\section{Utilization-Based Dynamic Posted-Price Mechanism}\label{sec:udpp}
To improve the robustness of the welfare-based dynamic posted-price mechanism to strategic behavior from miners, which in turn unravels the IC property for bidders, we can remove the influence of bid values on the price update when there is excess demand. That is, we can use a fixed multiplicative price change that is independent of bid values:
\begin{definition}[Utilization-Based Update Rule]\label{def:utilization-based} The {\em $(\alpha, \delta)$-utilization-based} update rule is the update rule
\begin{equation}
\tu(q,  B) := \alpha \frac{|B|}{m}(1+\delta)q + (1-\alpha) q.
\end{equation}
where $\alpha \in (0, 1)$ is the convergence parameter and $\delta \in (0, \infty)$.
\end{definition}
The utilization-based update rule with $\delta = 1$ and $\alpha = 1/8$ is the update rule used by EIP-1559 to dynamically set its posted price. While EIP-1559 uses a first-price auction to choose the allocation when there is excess demand, we adopt as the intended allocation rule the RM allocation rule which, in turn, removes incentives for non-truthful bidding by bidders:
\begin{definition}[Utilization-Based Dynamic Posted-Price Mechanism]\label{def:adpp}
The {\em utilization-based dynamic posted-price} (UDPP) mechanism is an RDPP mechanism $\left(\vec x, \vec p, \tu\right)$ endowed with the utilization-based update rule $\tu$.
\end{definition}
\begin{proposition}\label{prop:dpp-rm-ic-dsic}
The RDPP is ex post IC for myopic bidders and DSIC for myopic miners.
\end{proposition}
\begin{proof}
First, we show the RDPP mechanism is DSIC for myopic miners. For a particular time step, let $M$ denote the set of active bidders and let $q$ denote the current posted price. The active miner maximizes utility by selecting a maximal allocation from $M(q)$, and any maximal allocation weakly dominates all other allocations. Thus choosing a uniformly random maximal allocation is a weakly dominant strategy. This proves the RDPP mechanism is DSIC for myopic miners.

Next, we show the RDPP mechanism is {\em ex post} IC for myopic bidders. Given that miners select a uniformly random maximal allocation from $M(q)$, we first consider the case where bidder $i$ has value $v_i < q$. If bidder $i$ bids $b_i \geq q$, there is a positive probability the active miner selects bidder $i$ to receive a slot and pay $q$ resulting in a negative utility. Thus bidding $v_i$ strictly dominates bidding $b_i \geq q$ and weakly dominates bidding $b_i \in (v_i, q)$. Next, suppose bidder $i$ has value $v_i \geq q$. For bidder $i$ to receive a slot and pay $q$, she must bid $b_i \geq q$. Thus bidding $v_i$ strictly dominates bidding $b_i < q$. For any bid $b_i \geq q$, bidder $i$ has equal probability of winning a slot and paying $q$. Thus bidding $v_i$ weakly dominates bidding $b_i > v_i$. This proves truthful bidding weakly dominates bidding $b_i \neq v_i$, and the mechanism is IC for myopic bidders.
\end{proof}
\begin{corollary}\label{cor:adpp-ic-dsic}
The UDPP mechanism is {\em ex post} IC for myopic bidders and DSIC for myopic miners.
\end{corollary}
\vspace{1mm}\noindent\textit{Remark.} Proposition~\ref{prop:dpp-rm-ic-dsic} shows that a randomized dynamic posted-price mechanism shares the benefits of the tipless mechanism~\cite{roughgarden2020transaction} (both are {\em ex post} IC for myopic bidders and DSIC for myopic miners). However, they are vulnerable to off-chain agreements~\cite{roughgarden2010algorithmic} in that collusion between miners and bidders could improve the sum of their utilities: whenever the current posted price is smaller than the market-clearing price, the active miner can run an auction outside the blockchain to select the top $m$ transactions instead of randomizing the allocation. This deviation can improve the joint utility of miners and bidders since the RM allocation might allocate block space to a low value bidder in place of a higher value bidder. On the positive side, if the posted price is above the market-clearing price, off-chain agreements do not improve the joint utility of bidders and miners. Any additional revenue miners receive results in a higher payment from a bidder that would still receive block space if there were no off-chain agreement.

\vspace{1mm}\noindent\textit{Remark.} A dynamic posted-price mechanism with the welfare-based update rule and the RM allocation rule (Definition~\ref{def:rm}) is also DSIC for myopic miners. However, the crucial difference is that there is a clear and sensible out-of-model deviation by non-myopic miners to that rule, while this deviation has gone away in the new design.

Interestingly, the fact the UDPP mechanism ignores bid values when there is excess demand, which is important for incentive alignment, can lead to instability in instances in which the welfare-based dynamic-posted mechanism is asymptotically stable. This motivates the study of more stable update rules in Section~\ref{sec:twdpp}.\footnote{\citet{leonardos2021dynamical} provides additional examples of unstable behavior for the UDPP mechanism.}
\begin{proposition}\label{prop:instability}
There exists a distribution $F^2$ where the WDPP mechanism $\left(\vec x, \vec p, \tw\right)$ is asymptotically stable but the UDPP mechanism $\left(\vec x, \vec p, \tu\right)$ is unstable.
\end{proposition}
\begin{proof}
For any positive $v$, let $F$ be the distribution where each bidder has value $v$ with probability 1. Suppose there is a single slot for sale and two bidders drawn i.i.d. from $F$ in each time step. For the UDPP mechanism, if the current price $q_t \leq v$, the mechanism observes excess demand and the subsequent price increases to $q_t(1 + \alpha \delta)$. However, when the current price $q_t > v$, neither bidder will purchase, and the subsequent price decreases to $q_t(1-\alpha)$. This repeats {\em ad infinitum}, and the utilization-based update rule is unstable with respect to $F^2$.

For the welfare-based dynamic posted-price mechanism, whenever the current price $q_t > v$, neither bidder will purchase, and the subsequent price decreases to $q_t(1-\alpha)$. Thus eventually, the mechanism reaches a price $q_t \leq v$. Whenever the current price $q_t \leq v$, the active miner allocates a single slot and the subsequent price increases to $q_t + \alpha(v - q_t)$ where $v-q_t \geq 0$ is the surplus of the mechanism. Thus the subsequent price is $q_t < q_{t+1} = \alpha v + (1-\alpha)q_t \leq v$ whenever $q_t < v$ (and $q_{t+1} = v$ whenever $q_t = v$). This proves for all time steps $n \geq t$, $q_t \leq v$, and one bidder receives the slot. The sequence $q_t, q_{t+1}, \ldots$ is increasing and bounded by $v$. From completeness of the Euclidean space, the sequence has a limit point $z \leq v$. This proves the welfare-based update rule is asymptotically stable with respect to $F^2$. See Figure~\ref{fig:instability-1} for an empirical example of this behavior.
\end{proof}
By assuming the revenue curve is strictly concave and Lipschitz continuous we can derive a condition for the stability of the UDPP mechanism (and EIP-1559).
\begin{theorem}\label{thm:udpp-stability}
Let $F$ be a distribution over bidder valuations where values are bounded by $\bar a$ and $F$ induces the revenue curve $\rev(q)$ to be $L$-Lipschitz and strictly concave. Then if $\alpha \leq \frac{1}{L(1+\delta)/m+1}$, the UDPP mechanism $\left(\vec x, \vec p, \tu\right)$ is asymptotically stable with respect to $F$ and has unique equilibrium price.
\end{theorem}
The proof will require the same techniques of the proof for the stability of the TWDPP mechanism. We defer the proof to Section~\ref{sec:twdpp-stability}.
\section{The Truncated Welfare-Based Dynamic Posted-Price Mechanism}\label{sec:twdpp}
In this section, we propose the {\em truncated welfare-based dynamic posted-price} (TWDPP) mechanism that combines desirable features of the WDPP and the UDPP mechanisms. In particular, whenever $m$ or more bidders receive a slot, TWDPP behaves like the UDPP mechanism. Whenever fewer than $m$ bidders receive a slot, we use bid values to compute the subsequent price just like the WDPP mechanism, but to reduce the variance of the update rule, we  cap how much each bid contributes to the subsequent price.
\begin{definition}[Truncated Welfare-Based Update Rule]\label{def:truncated-rule}
The {\em truncated $(\alpha, \delta)$-welfare-based} update rule is the update rule
\begin{equation}\label{eq:ttw}
\ttw(q, B) := \begin{cases}
\alpha \frac{1}{m} \sum_{i \in B} \min\{b_i, (1+\delta) \cdot q\} + (1-\alpha) \cdot q & \text{for $|B| < m$,}\\
\alpha (1+\delta) \cdot q + (1-\alpha)q & \text{for $|B| = m$,}
\end{cases}
\end{equation}
where $\alpha \in (0, 1)$ is the convergence parameter and $\delta \in (0, \infty)$ is the truncation parameter.\footnote{An alternative update rule could always round $b_i$ to the interval $[(1-\delta)q, (1+\delta)q]$ so that $T(q, B) = \alpha\frac{1}{m}\sum_{i \in B} \max\{\min\{b_i, (1+\delta)q\}, (1-\delta)q\} + (1-\alpha)\cdot q$. Here we study the update rule in Equation~\ref{eq:ttw} to exclude out-of-model deviations where miners prefer to include higher value bids whenever $|B| = m$.}
\end{definition}

\begin{definition}[Truncated Welfare-Based Dynamic Posted-Price Mechanism]\label{def:twdpp}
The {\em truncated welfare-based dynamic posted-price} (TWDPP) mechanism is an RDPP mechanism $(\vec x, \vec p, \ttw)$ endowed with the truncated welfare-based update rule $\ttw$.
\end{definition}
Because the TWDPP mechanism is an RDPP mechanism, Proposition~\ref{prop:dpp-rm-ic-dsic} implies:
\begin{theorem}\label{thm:twdpp-ic-dsic}
The TWDPP mechanism is {\em ex post} IC for myopic bidders and DSIC for myopic miners.
\end{theorem}
Recall from Proposition~\ref{prop:instability} that there is a distribution where the UDPP mechanism is unstable, but the WDPP mechanism is stable. Since the TWDPP mechanism behaves like the UDPP mechanism when blocks are full, the TWDPP mechanism can be unstable when the WDPP mechanism is stable (simply consider the same distribution constructed in Proposition~\ref{prop:instability}). However, the following proposition shows the TWDPP mechanism can still provide advantages over the UDPP mechanism by constructing a distribution $F^n$ where the UDPP mechanism is unstable even when the demand $n$ is smaller than the supply $m$, but the TWDPP mechanism is stable.
\begin{proposition}\label{prop:twdpp-stable}
There exists a distribution $F^n$ with $n < m$ where the TWDPP mechanism $\left(\vec x, \vec p, \ttw\right)$ is asymptotically stable, but the UDPP $\left(\vec x, \vec p, \tu\right)$ mechanism is unstable.
\end{proposition}
We defer the proof to the appendix. To study the stability of the TWDPP mechanism, we define the expected value of the random operator $\ttw(q, \unif(2^{M(q)}))$:
\begin{equation}\label{eq:ettw}
E_{\ttw}(q) := \alpha \ktw(q) + (1-\alpha)q,
\end{equation}
where $\ktw$ is the kernel of the TWDPP mechanism:
\begin{align}\label{eq:ktw}
\ktw(q) &:= \mathbbm E\bigg[\frac{1}{m}\sum_{i = 1}^n \min\{v_i, (1+\delta)q\}\cdot x_i(q) \cdot \ind{N(q) < m}\\
&\qquad\qquad + (1+\delta)\cdot q \cdot \ind{N(q) \geq m}\bigg],
\end{align}
where the expected value is taken over the values $v_1, v_2, \ldots, v_n$ and random allocation $\unif(2^{M(q)})$. Recall $x_i(q) = 1$ if bidder $i$ with $v_i \geq q$ receives a slot and $N(q) = \sum_{i = 1}^n \ind{v_i \geq q}$.
\subsection{Stability}\label{sec:twdpp-stability}
Because the kernel $\ktw$ is neither increasing nor decreasing, we cannot directly apply the proof sketch in Section~\ref{sec:fixed-point} to show the TWDPP mechanism is asymptotically stable. Instead, we will show that fixed-point iteration on $E_{\ttw}$ starting from a positive $q_0$ converges to a fixed point under the assumption that $\ktw$ is Lipschitz continuous and strictly concave.
\begin{definition}[Concave function]\label{def:concave}
A real-valued function $f : X \to \real$ is {\em concave} if for any $x \neq y \in X$ and $0 < \alpha < 1$, $f((1-\alpha)x + \alpha y) \geq (1-\alpha)f(x) + \alpha f(y)$, and {\em strictly concave} if we have a strict inequality.
\end{definition}
Here, we show that for a function $f$ satisfying Assumption~\ref{assumption:concave-1}, the fixed-point iteration on $g(x) = \alpha f(x) + (1-\alpha)x$ (for some positive $\alpha < 1$) starting from a positive $x_0$ converges to a fixed point of $f$.\footnote{It is known that iterating an increasing function $f : \real_{\geq 0} \to \real_{\geq 0}$ will either converge to a fixed point or diverge~\cite{kennan2001uniqueness}. In Section~\ref{sec:fixed-point}, we showed that if $f$ is decreasing and Lipschitz continuous, then the fixed-point iteration of $g(x) = \alpha f(x) + (1-\alpha)x$ will converge to a unique fixed point by setting a sufficiently small positive $\alpha$. Here, we replace the monotonicity assumption by the assumption $f$ is strictly concave.}
\begin{assumption}\label{assumption:concave-1}
Suppose $X = [0, \bar a] \neq \emptyset$ is a finite interval and $f : \mathbb R_{\geq 0} \to \mathbb R_{\geq 0}$ is a real-valued function satisfying:
\begin{enumerate}[label=(\roman*)]
    \item \label{concave-2-0} $f$ is $L$-Lipschitz.
    \item \label{concave-2-1} The restriction of $f$ to $X$ is strictly concave.
    \item \label{concave-2-2} $f(x) = 0$ for all $x \geq \bar a$.
\end{enumerate}
\end{assumption}
\begin{assumption}\label{assumption:concave-2}
Assume $f : X \to X$ with nonempty interval $X = [0, \bar a]$ is a real-valued function satisfying:
\begin{enumerate}[label=(\roman*)]
\item \label{concave-0} $f$ is strictly concave and Lipschitz continuous.
\item \label{concave-1} $f(x) \geq 0$ for all $x \in X$.
\item \label{concave-3} Let $\bar x = \sup_{x \in X} f(x)$. Strict concavity implies $f$ is strictly increasing on interval $I = [0, \bar x)$ and strictly decreasing on interval $D = [\bar x, \sup X]$. Then assume the restriction of $f$ to $D$ is $L$-Lipschitz with $L < 1$.
\item \label{concave-4} There is a positive $b \in X$ such that $f(b) < b$. 
\end{enumerate}
\end{assumption}
The main idea is to reduce the problem of computing a fixed point of $f$ satisfying Assumption~\ref{assumption:concave-1} to finding a fixed point for another function $g$ satisfying Assumption~\ref{assumption:concave-2} such that $x$ is a fixed point of $f$ if and only if $x$ is a fixed point of $g$. The main observation is that whenever $f : \real_{\geq 0} \to \real_{\geq 0}$ is nonincreasing on some interval $D \subseteq \real_{\geq 0}$, we can force the restriction of $g$ to $D$ to be Lipschitz continuous with Lipschitz constant $L < 1$ by using a similar argument from Lemma~\ref{lemma:contractive-update-rule}. We can leverage this fact to show that the restriction of $g$ to some interval $[a, b] \subseteq D$ is a contraction---$g$ maps $[a, b]$ to itself and the restriction of $g$ to $[a, b]$ has Lipschitz constant $L < 1$. Then, we show that iterating $g$ starting from any positive $x_0$ will either converge to a point outside $[a, b]$ or eventually visit a point $x_t \in [a, b]$. The moment $x_t \in [a, b]$ for some $t$, from Banach contraction principle and the fact the restriction of $g$ to $[a, b]$ is a contractive mapping implies the fixed-point iteration of $g$ starting from $x_t$ will converge to a unique fixed point contained in $[a, b]$.
\begin{lemma}\label{lemma:concave-contraction}
If $f : X \to X$ satisfies Assumption~\ref{assumption:concave-2}, then the fixed-point iteration of $f$ starting from any $x_0 \in X$ converges to a fixed point of $f$.
\end{lemma}
\begin{lemma}\label{lemma:reduction}
Let $f : \real_{\geq 0} \to \mathbb R_{\geq 0}$ be a function satisfying Assumption~\ref{assumption:concave-1}. By setting $0 \leq \alpha \leq \frac{1}{L+1}$, the function $g(x) = \alpha f(x) + (1-\alpha) x$ satisfies Assumption~\ref{assumption:concave-2}. Then, the fixed-point iteration of $g$ starting from any positive $x_0$ converges to a fixed point of $f$.
\end{lemma}
We defer the proof of the Lemmas to Appendix~\ref{sec:twdpp-stability-appendix}. In Example~\ref{example:reduction}, we run simulations demonstrating the fixed-point iteration described in Lemma~\ref{lemma:reduction}. 

Lemma~\ref{lemma:reduction} immediately implies that if the kernel $\ktw$ is Lipschitz continuous and strictly concave on a bounded interval $[0, \bar a]$, the TWDPP mechanism is asymptotically stable. We can now state our main theorem.
\begin{theorem}\label{thm:twdpp-stability}
Let $F$ be a distribution over bidder valuations where all buyers have value at most $\bar a$. Assume $F$ induces the kernel $\ktw$ to be $L$-Lipschitz and strictly concave on interval $[0, \bar a]$. Then for any $\alpha \leq \frac{1}{L + 1}$, the TWDPP mechanism is asymptotically stable with respect to $F$.
\end{theorem}
\begin{proof}
We first check that $\ktw$ satisfies all properties in Assumption~\ref{assumption:concave-1}. The first and second properties are clear. For the third property, note that $\ktw(q) \geq 0$, and observe that no bidder purchases at price $> \bar a$. Thus $1-F(q) = 0$ and $\ktw(q) = 0$ for all $q \geq \bar a$. This proves $\ktw$ satisfies Assumption~\ref{assumption:concave-1}. From Lemma~\ref{lemma:reduction}, for any $\alpha \leq \frac{1}{L+1}$, the fixed-point iteration of
$$E_{\ttw}(q_0) = \alpha \ktw(q_0) + (1-\alpha)q_0$$
starting from positive $q_0$ converges to a fixed point $q^*$. This proves the intended TWDPP mechanism is asymptotically stable.
\end{proof}
The proof of the stability of the UDPP mechanism is identical to the proof of the stability of the TWDPP mechanism.
\begin{proof}[Proof of Theorem~\ref{thm:udpp-stability}]
Note the expected value of the price update rule $\tu(q, B)$ is
$$E_{\tu}(q) = \alpha \frac{\rev(q) (1+\delta)}{m} + (1-\alpha) q.$$
Observe $\rev(q) (1+\delta)/m$ is $\frac{L(1+\delta)}{m}$-Lipschitz and satisfies all conditions in Assumption~\ref{assumption:concave-1}. From Lemma~\ref{lemma:reduction}, $E_{\tu}(q)$ is asymptotically stable whenever $\alpha \leq \frac{1}{L(1+\delta)/m+1}$.
\end{proof}
\subsection{Welfare guarantee at equilibrium}\label{sec:twdpp-welfare}

Because the active miner selects a random maximal allocation whenever $N(q) > m$, it is more challenging to show the TWDPP mechanism is approximately welfare optimal. That is because the social welfare is not a monotone function of the posted price. As the posted price decreases, the welfare might decrease because the active miner has a higher chance to allocate block space to a lower bid instead of a higher bid. 

The main idea is to show that if $q$ is an equilibrium price, the probability that $N(q) \geq m$ is at most $1/(1+\delta)$. Whenever $N(q) \leq m$, miners deterministically allocate slots to the top $N(q)$ bidders. Thus we will show the welfare loss from the cases where $N(q) > m$ is a constant fraction of the optimal welfare.
\begin{theorem}\label{thm:twdpp-welfare}
If $q$ is an equilibrium price of the TWDPP mechanism $\left(\vec x, \vec p, \ttw\right)$, then
$$\welfare(\vec x(q)) \geq \frac{\opt}{2(1+\delta)}\min\{1, \delta\}.$$
\end{theorem}
We defer the proof to the appendix. By setting $\delta = 1$, we get that the TWDPP mechanism obtains 1/4 of the optimal welfare at equilibrium.

\section{Experimental Results}\label{sec:experimental-results}

In this section, we describe some experimental results supporting our theoretical findings. We defer additional experiments to Appendix~\ref{sec:other-experimental-results}. We implement an agent-based simulator (included in the supplemental materials) with which we simulate the following dynamic posted-price mechanisms:
\begin{itemize}
    \item \emph{Welfare-Based Dynamic Posted-Price} (WDPP) mechanism $\left(\vec x, \vec p, \tw\right)$: the deterministic dynamic posted-price mechanism with the maximum value allocation rule and the welfare-based update rule.
    \item \emph{Utilization-Based Dynamic Posted-Price} (UDPP) mechanism $\left(\vec x, \vec p, \tu\right)$: the randomized dynamic posted-price mechanism with the random maximal allocation rule and the utilization-based update rule.
    \item \emph{Truncated Welfare-based Dynamic Posted-Price} (TWDPP) mechanism $\left(\vec x, \vec p, \ttw\right)$: the randomized dynamic posted-price mechanism with the random maximal allocation rule and the truncated welfare-based update rule.
\end{itemize}
While UDPP and TWDPP mechanisms are {\em ex post} IC, the WDPP mechanism is not. However, for the simulations, we always assume truthful bidding (in this sense, the results are overly optimistic for WDPP). Recall UDPP mechanism is equivalent to the EIP-1559 protocol in the case that a bidder $i$ bids $b_i = 0$ and submits a bid cap $c_i = v_i$.

For the experimental setup, at each time step $t$, $n_t$ bidders are drawn i.i.d. from distribution $F_t$. We use the same random seed across all mechanisms. We set the convergence parameter $\alpha = 1/16$. We set the truncation parameter $\delta = 1$. Decreasing $\delta$ would reduce how much each bid contributes to increasing the subsequent price.

Throughout the simulations, we track three quantities for each time step $t$: the demand $n_t$, the posted price $q_t$, and the ratio between the achieved welfare $\sum_{i \in B_t} v_i$ and the optimal welfare $\sum_{i = 1}^{n_t} v_i \cdot x_i^*(\vec v_t)$, where $\vec v_t$ is drawn from  $F_t^{n_t}$, the $n_t$-dimensional product distribution on $F_t$.

\subsection{Excess demand}\label{sec:excess-demand}

As an illustration, consider the case with constant demand $n_t = 200$ above the supply of $m = 100$ and three different value distributions: a uniform distribution over $0$ to $200$, an exponential distribution with a mean of 100, and a scaled and shifted Pareto distribution with a median around 100. See Figure~\ref{fig:constant-demand}. 

The WDPP mechanism has the fastest convergence for distributions with low variance, but is highly volatile when the distribution has high variance (e.g., a Pareto distribution). This kind of high price volatility in WDPP  would also raise doubts about the assumption of bidder myopia since they could benefit by timing their bids. Because the UDPP mechanism ignores bid values, it converges to a higher average price than the TWDPP mechanism, with the effect of providing lower welfare. This phenomenon can be observed clearly for the exponential distribution. The TWDPP mechanism more robustly handles both the low- and high-variance cases, consistently achieving the best welfare of around $4/5$ of the optimal welfare (outperforming the worst-case $1/4$ bound).

\begin{figure*}
\centering
\includegraphics[width=\textwidth]{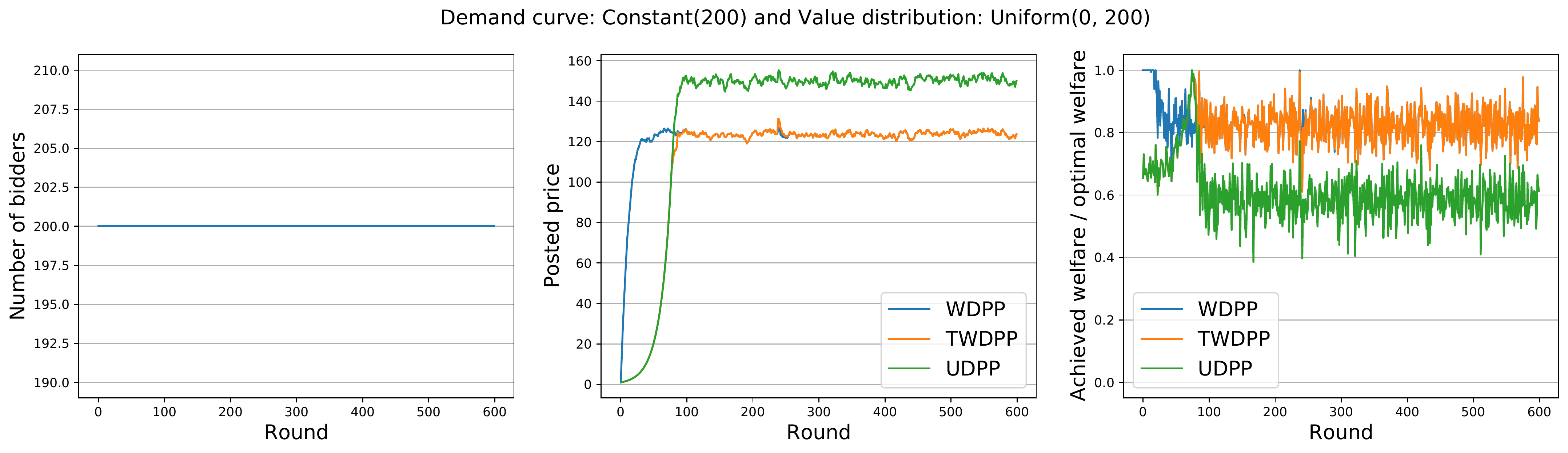}
\includegraphics[width=\textwidth]{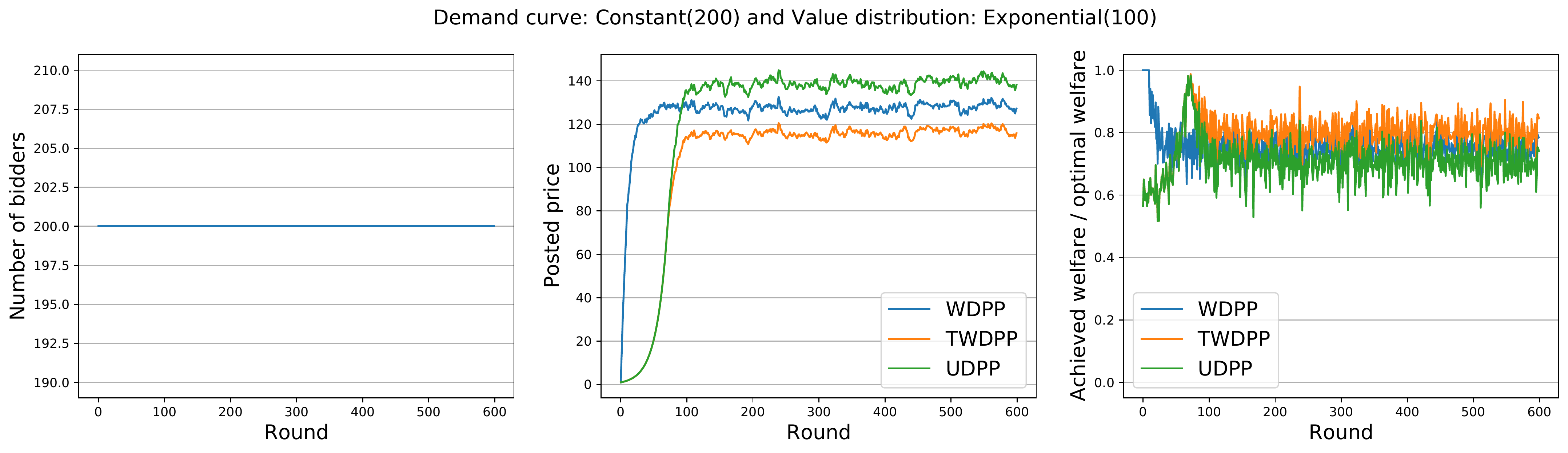}
\includegraphics[width=\textwidth]{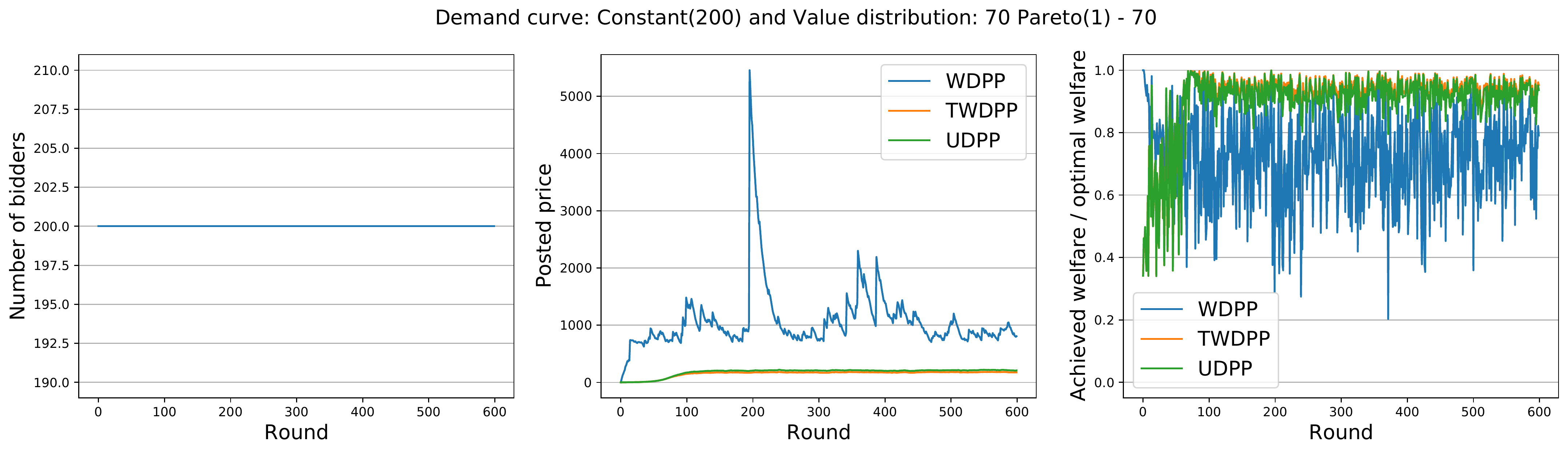}
\caption{Simulations with constant demand $n_t = 200$ for all $t$, constant supply $m = 100$, and uniform, exponential, and Pareto value distributions.}
\label{fig:constant-demand}
\end{figure*}

\section{Conclusion}\label{sec:conclusion}

First-price auctions have been the standard transaction fee mechanism for blockchains since Bitcoin's inception. While incentive-aligned for miners, first-price auctions provide a bad user experience for users under high competition when it is difficult to estimate how much to bid to get a transaction accepted. Moreover, typical alternatives from mechanism design theory do not work in the blockchain setting because a miner may deviate from the intended auction allocation rules (with the effect of also hiding information from other participants).

The question, then, is how to design a transaction fee mechanism that provides a good user experience and is at the same time robust to manipulation by miners. The Ethereum community has proposed and recently adopted EIP-1559, which is a posted-price and first-price hybrid. EIP-1559 reduces to the incentive-compatible UDPP mechanism that we study when the demand at the current posted price is lower than the supply, but it is not incentive compatible otherwise.

We have proposed the {\em Truncated Welfare-based Dynamic Posted-Price (TWDDP) mechanism}. This mechanism is incentive compatible for myopic bidders and aligns incentives for myopic miners to follow the rules, even during periods of market instability. It converges to a stable price under strict concavity and Lipschitz continuity assumptions on a quantity derived from the distribution of bidder valuations. At this stable price, the mechanism is also approximately welfare optimal. Our experimental results, conducted in simulation, confirm our theoretical findings.

\subsection{Future work}

This line of work leaves many interesting open questions. As with previous work, we study transaction fee mechanisms under the assumption that miners are myopic. However, as blockchain mining becomes more centralized through mining pools, miners may become motivated to take on non-myopic strategies. Thus it is interesting to understand the strategic behavior of non-myopic miners, seeking bounds on gains available to miners for deviating from the intended allocation rule as well as rules that are robust even to forward-looking miners.

Another interesting direction is to consider the effect of new transaction fee mechanisms on miner incentives to follow the consensus protocol.  Block withholding, for example, where a miner hides one or more blocks in an attempt to fork honest blocks from the blockchain, is  a forward-looking strategy that can be effective for miners who are motivated by block rewards~\cite{eyal2014majority, sapirshtein2016optimal, ferreira2021proof} or transaction fees~\cite{carlsten2016instability}. The redesign of the transaction fee mechanism modifies the incentives of miners, and leaves the interesting open question of understanding the impact of this change on block withholding strategies.

Lastly, while we model single-dimensional auctions in which bidders have distinct private values but identical transaction sizes, the transaction fee auctions that take place within Ethereum more closely resemble {\em knapsack auctions}, where different transactions can occupy different amounts of block space. Thus, a relevant question for transaction fee markets is that of the design of credible and communication-efficient mechanisms for multi-dimensional settings such as knapsack auctions~\cite{akbarpour2020credible, ferreira2020credible, daskalakis2020simple}.

\bibliographystyle{plainnat}
\bibliography{refs}

\newpage
\appendix

\section{Other Experimental Results}\label{sec:other-experimental-results}

\subsection{Under demand}\label{sec:under-demand}
Here we replicate the experiments from Section~\ref{sec:excess-demand} but with demand $n_t = 67$ smaller than the supply $m = 100$. When the demand is lower than the supply, the market-clearing price is $0$. In general, we cannot expect our mechanisms to converge to a posted price of $0$, but we can hope the mechanisms converge to a price that sells to all almost all bidders. We give the results in Figure~\ref{fig:constant-demand-2}.

Unfortunately, the WDPP mechanism continues to have high price volatility, which contributes to low welfare on average. Both the UDPP and the TWDPP mechanisms sell to almost all bidders, with the TWDPP having a slightly higher welfare because it converges to a lower average price.

\begin{figure}[H]
\centering
\includegraphics[width=\textwidth]{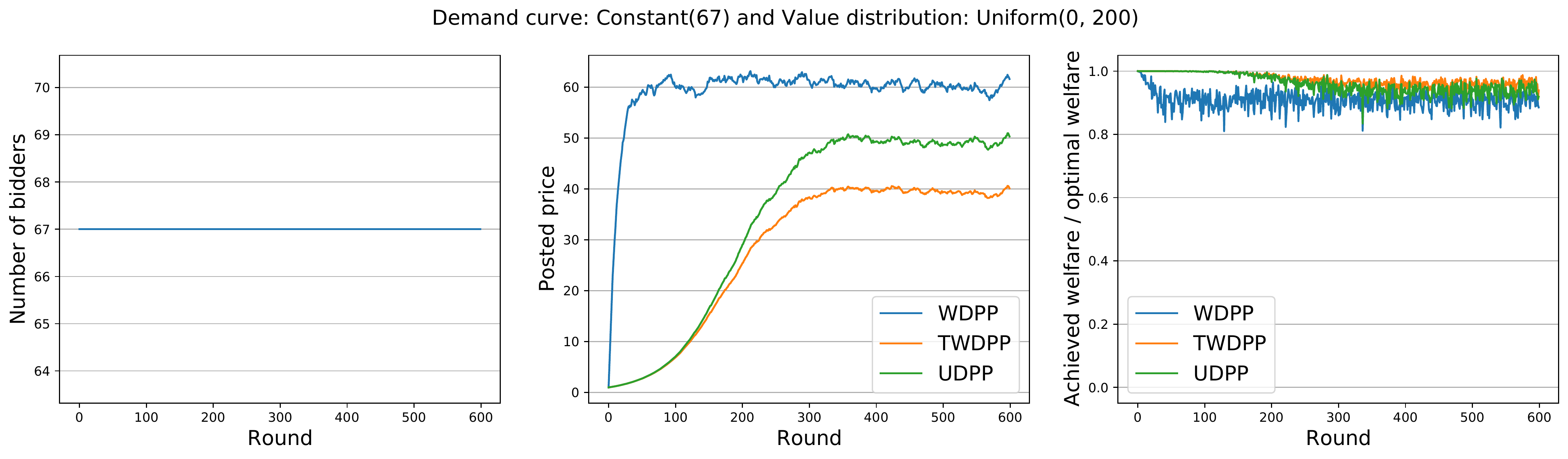}
\includegraphics[width=\textwidth]{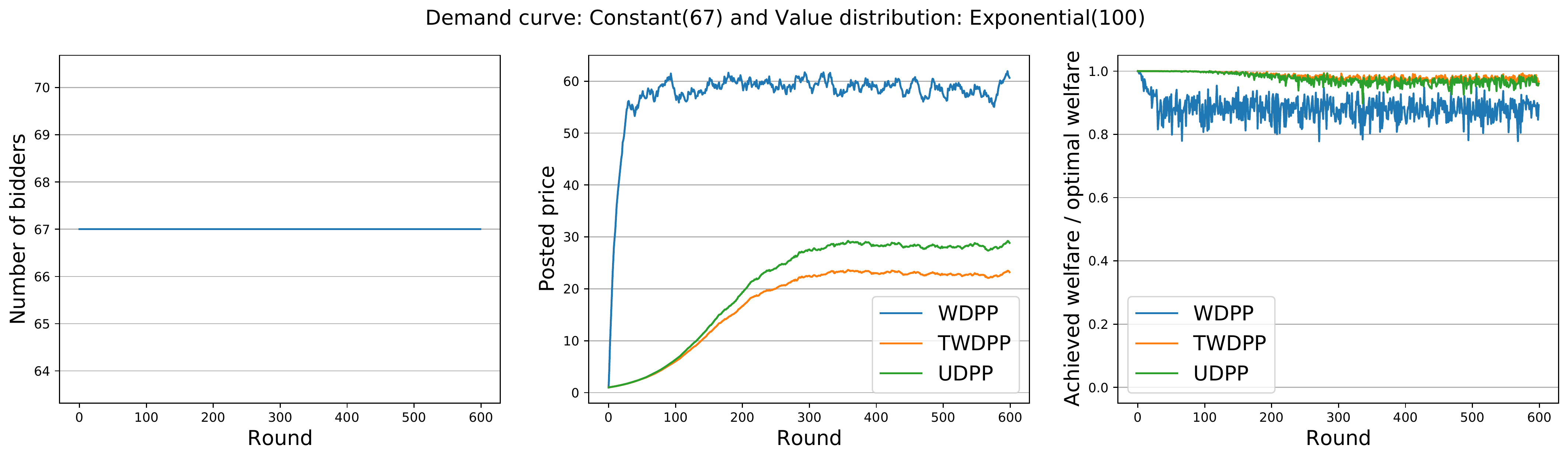}
\includegraphics[width=\textwidth]{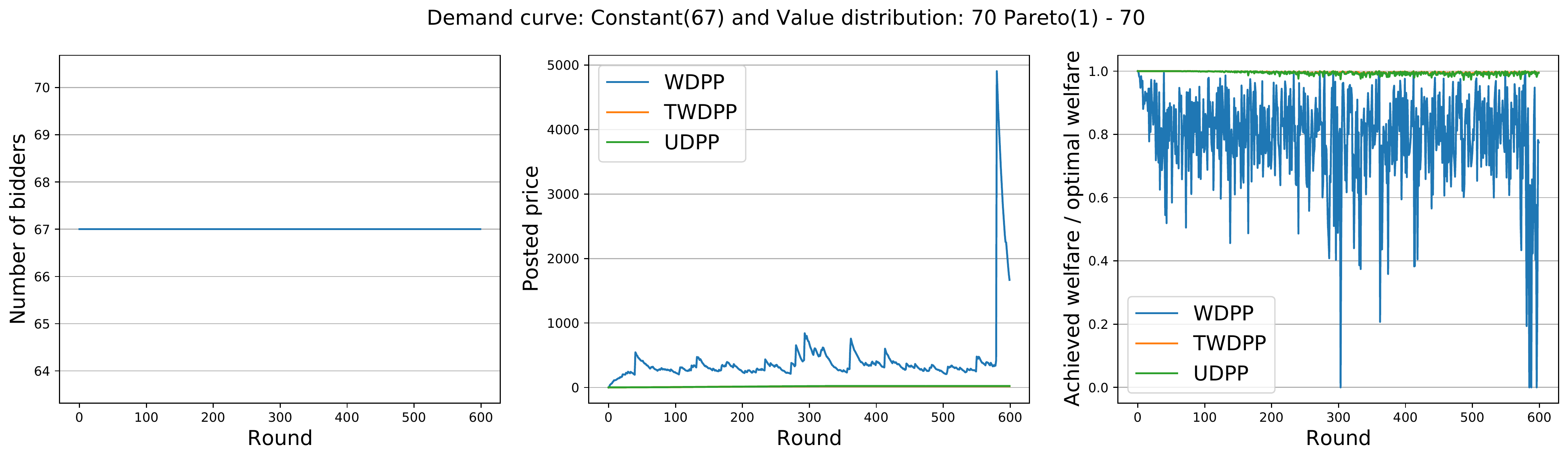}
\caption{Simulations with constant demand $n_t = 67$ for all $t$, constant supply $m = 100$, and uniform, exponential, and Pareto value distributions.}
\label{fig:constant-demand-2}
\end{figure}

\newpage
\subsection{Demand shocks}\label{sec:ex-shock}
We repeat the experiments from Section~\ref{sec:excess-demand} with a demand shock: the demand starts at $200$ transactions per time step, jumps to $600$, and later returns to $200$. The results are given in Figure~\ref{fig:additional-stepwise-demand}.

We observe that all mechanisms converge quickly to an equilibrium price under demand shock. During demand shock, the WDPP mechanism's instability increases for distributions with high variance (e.g., the Pareto distribution). The TWDPP mechanism has strictly higher welfare for all value distributions because it converges to an average price closer to the market-clearing price.

We observe that the TWDPP mechanism can have price volatility when values are drawn from the uniform distribution over $[0, 200]$. That happens because whenever a block is full, the TWDPP mechanism uses the utilization-based update rule. Thus when the TWDPP mechanism reaches a price $q$ close to 175, there is a high probability that there are at least $m$ bidders with values above $q$ causing the subsequent increase to $q(1+\alpha)$. However, this instability is unpredictable (differently from the instability we will observe in Section~\ref{sec:instability}) and even under instability, the TWDPP mechanism extracts higher welfare than the UDPP mechanism.

\begin{figure}[H]
\centering
\includegraphics[width=\textwidth]{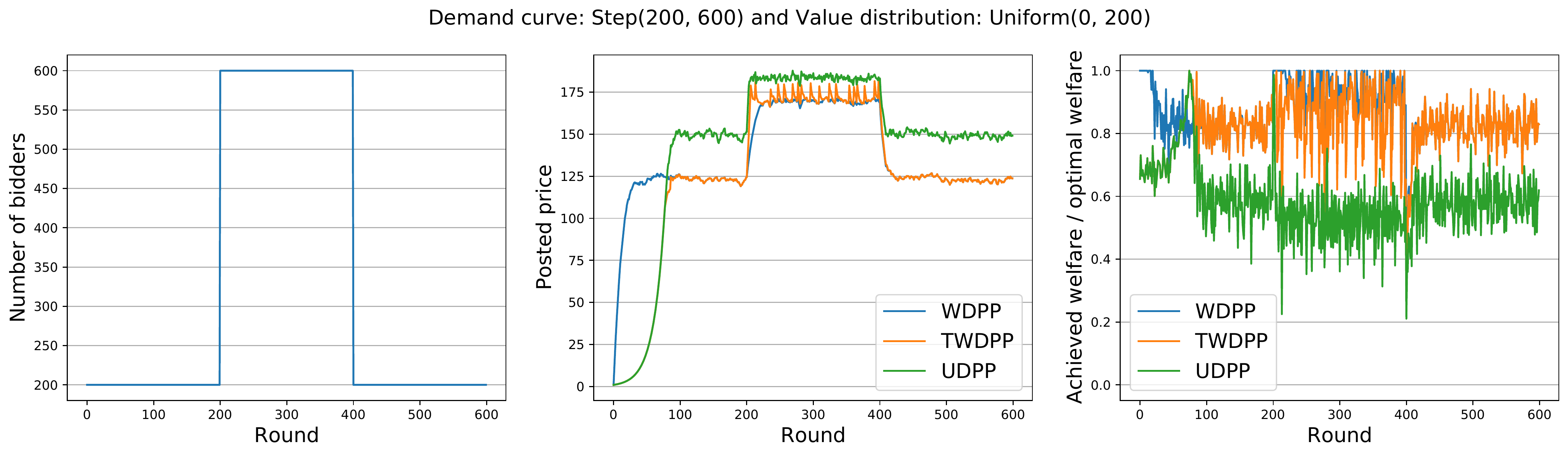}
\includegraphics[width=\textwidth]{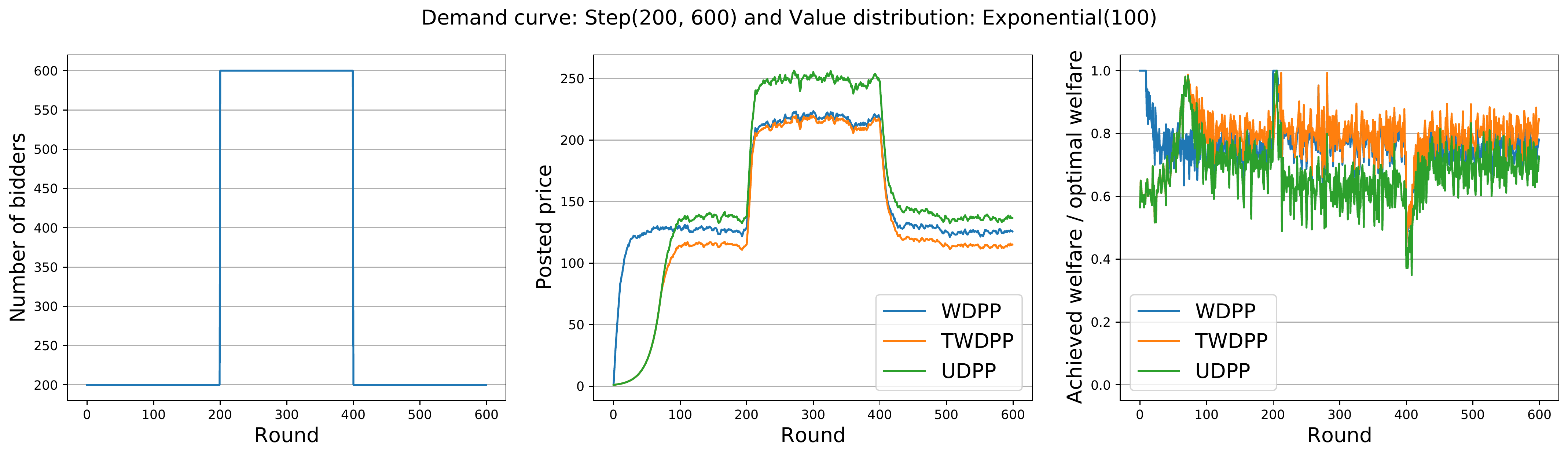}
\includegraphics[width=\textwidth]{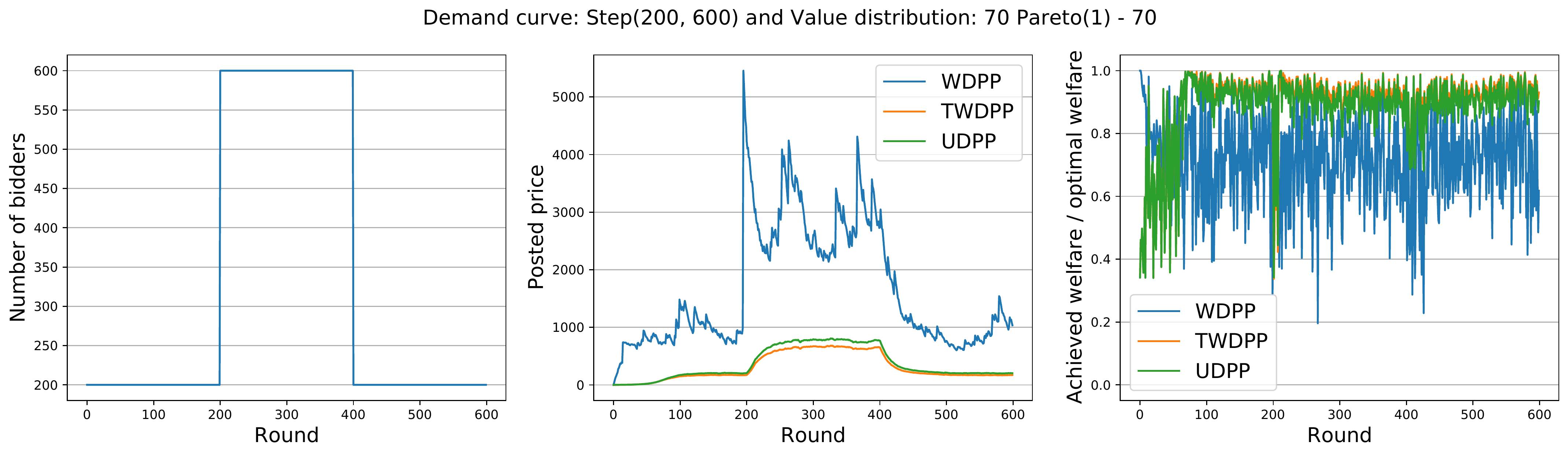}
\caption{The demand follows a step function that starts at $200$, jumps to $600$, and finally returns to $200$. As with the previous simulations, we use a fixed supply of $m = 100$ slots per time step, and we provide results for uniform, exponential, and Pareto value distributions.}
\label{fig:additional-stepwise-demand}
\end{figure}
\subsection{Instability of EIP-1559 utilization-based update rule}
\label{sec:instability}
This section considers the case where all bidders have a deterministic value of $100$. The experiment in Figure~\ref{fig:instability-1} serves as empirical support for Proposition~\ref{prop:instability}. It provides an example of a distribution where the WDPP mechanism is asymptotically stable, but the utilization-based update rule from EIP-1559 (and hence the UDPP mechanism) is unstable.

From the experiment, we observe the WDPP mechanism converges to a stable price of $100$ while the UDPP mechanism oscillates around $100$. The price oscillation weakens the bidder's myopic assumption. In particular, a patient bidder would wait to enter a bid only when the price is below 100. For this example, one can check the TWDPP mechanism is also unstable because the truncated welfare-based update rule reduces to the welfare-based update rule.

\begin{figure}[H]
\centering
\includegraphics[width=\textwidth]{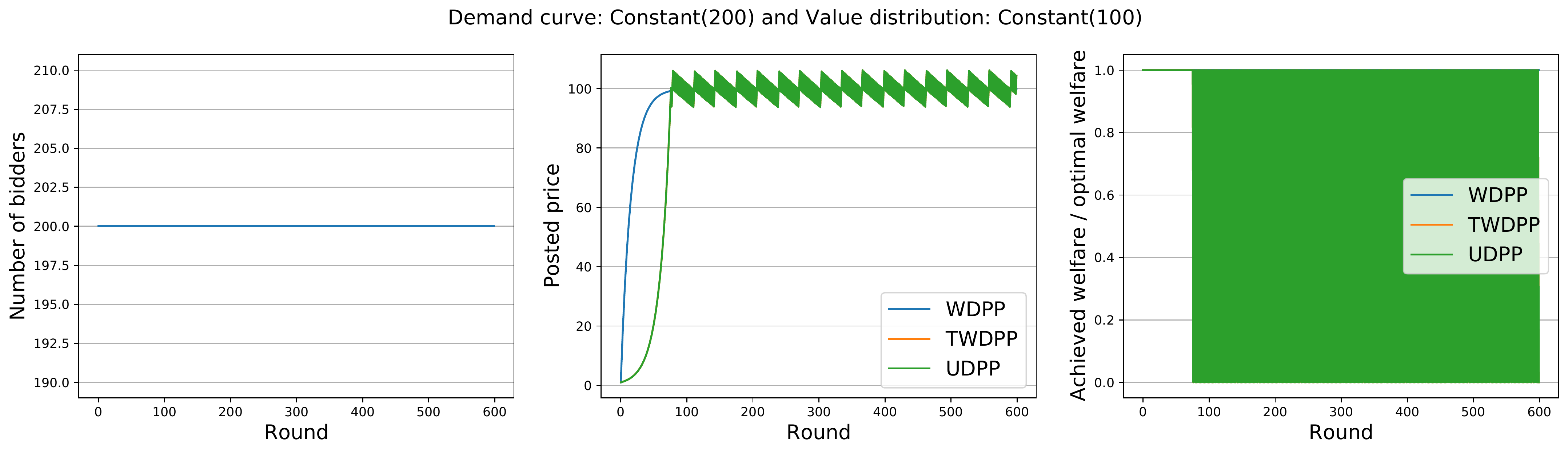}
\caption{Simulation with constant demand $n = 200$, constant supply $m = 100$, and deterministic bidder value $v = 100$.}
\label{fig:instability-1}
\end{figure}
Next, we modify the example in Figure~\ref{fig:instability-1} and consider the case where the demand is smaller than the supply. The experiment in Figure~\ref{fig:instability-2} serves as empirical support for Proposition~\ref{prop:twdpp-stable}. It provides an example of a distribution where the WDPP and TWDPP mechanisms are asymptotically stable, but the UDPP mechanism is unstable even when demand is lower than the supply.
\begin{figure}[H]
    \centering
    \includegraphics[width=\textwidth]{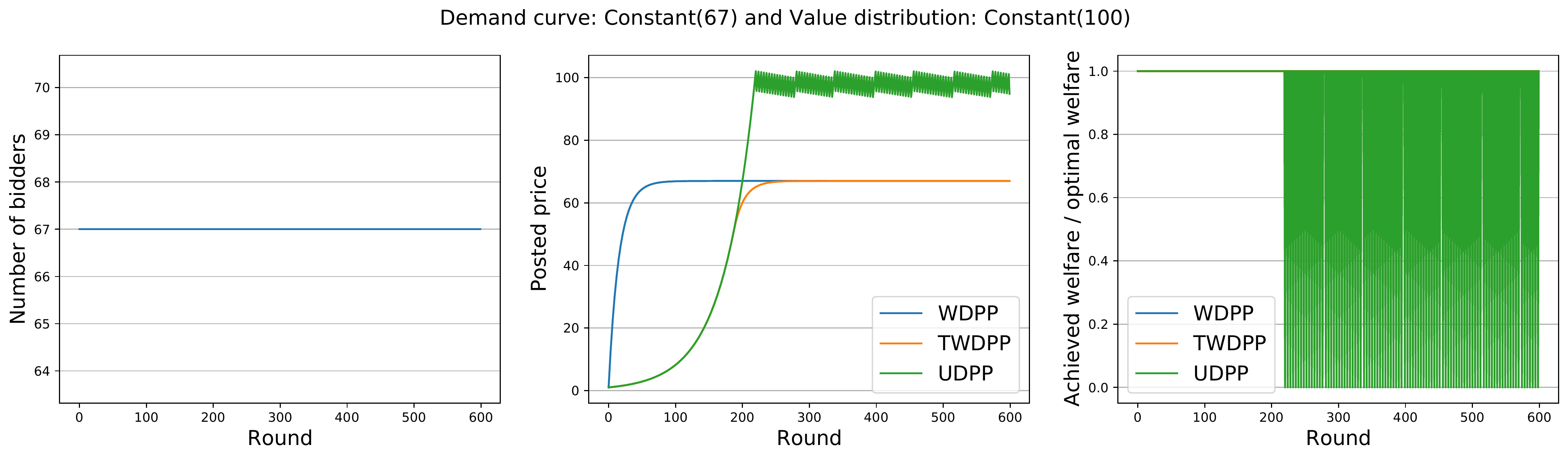}
    \caption{Simulation with constant demand $m = 67$, constant supply $m = 100$, and deterministic bidder value $v = 100$.}
    \label{fig:instability-2}
\end{figure}

\section{Existing Proposals of Blockchain Transaction Fee Mechanisms}\label{sec:summary-mechanisms}

\subsection{EIP-1559}\label{sec:eip-1559}
First-price auctions (Definition~\ref{def:first-price}) have been the standard choice for transaction fee mechanisms since \citet{nakamoto2008bitcoin} first proposed using them for Bitcoin. \citet{buterin2019eip1559}'s EIP-1559 proposal for Ethereum is closest to ours since their mechanism is endowed with an update rule for their posted price. However, their proposal contains a first-price auction component and a posted-price component, and here we study purely posted-price mechanisms.

\begin{definition}[EIP-1559~\cite{buterin2019eip1559}]\label{def:eip} %
The mechanism proposed in EIP-1559 is a dynamic mechanism $\left(\vec x, \vec p,\tu\right)$ endowed with the utilization-based update rule $\tu$ (Definition~\ref{def:utilization-based}) and a initial posted price $q_0$. Each bidder submits a tuple $(b_i, c_i)$ where $b_i$ is the bid and $c_i$ is a bid cap. Given history $B_1, B_2, \ldots, B_{t-1}$ and active bidders $M_t$, the mechanism computes the posted price $q_t = T(q_{t-1}, B_{t-1})$ for time step $t$. The intended allocation rule allocates slots to the top $m$ bids $b_i$ with $b_i + q_t \leq c_i$. For all bidders that receive a slot, the intended payment rule charges $b_i + q_t$, but miners are paid only $b_i$. Thus the utility of the active miner is
$$u_0^{\text{EIP-1559}}(q_t, B_t) := \sum_{i \in B_t} (p_i(B_1, \ldots B_t) - q_t).$$
\end{definition}
For EIP-1559, \citet{buterin2019eip1559} propose setting $\delta = 1$ in the hope the mechanism converges to an equilibrium that sells $m/2$ of the slots. Formally studying the stability of EIP-1559 is more challenging since truthful bidding is not an {\em ex post} Nash equilibrium, and any formal analysis would likely require analyzing bidders participating in a Bayesian Nash equilibrium. An alternative approach is to study EIP-1559 without the first-price auction component. The resulting mechanism is equivalent to the UDPP mechanism we study in Section~\ref{sec:udpp}.

\subsection{Monopolistic price auction}\label{sec:monopolistic}
We study mechanisms that sell a fixed supply of block slots. Nevertheless, since blockchain space is a digital good~\cite{goldberg2006competitive}, it is not obvious why the supply must be fixed a priori. With that in mind, \citet{lavi2019redesigning} propose using the \emph{monopolistic price auction} as a transaction fee mechanism, allowing miners to sell an unlimited supply of block slots. They showed the monopolistic price auction is DSIC for myopic miners and approximately IC for myopic bidders when values are drawn i.i.d. from a bounded distribution $F = [\underbar a, \bar a]$ and demand $n$ gets large -- i.e., $n \to \infty$.

\begin{definition}[Monopolistic price auction~\cite{lavi2019redesigning}]\label{def:monopolistic} The \emph{monopolistic price auction} is a static mechanism where the allocation rule allows an unlimited supply of slots to be allocated. First, the mechanism orders the active bids in decreasing order $b_1 \ge b_2 \ge \ldots \ge b_n$ and lets
$$k = \arg\max_i i \cdot b_i.$$
Then, it allocates slots to bidders $1$ through $k$. All bidders that receive a slot pay $b_k$ to the active miner.
\end{definition}

The monopolistic price is DSIC for myopic miners, and \citet{yao2018incentive} showed it is also approximately IC for myopic bidders for any distribution $F$ when the demand $n \to \infty$. Unfortunately, there is an inherent trade-off between having larger blocks and loss in social welfare.\footnote{Ethereum allows the blockchain protocol to converge to a fixed supply $m$ that is agreed upon by the majority of miners by allowing each active miner to upvote or downvote the supply by $\approx 0.1\%$. Thus one might expect $m$ to converge to a block size that balances the trade-offs between latency and supply.} Large block sizes result in higher network propagation delay, which increases the chance of forks. As a result, an unlimited supply of slots brings with it the possibility of new denial-of-service attacks.

\subsection{Random sampling optimal price}\label{sec:rsop}
\citet{goldberg2006competitive} proposed the \emph{random sampling optimal price} (RSOP) mechanism as a randomized truthful mechanism for selling an unlimited supply of items. The RSOP mechanism randomly partitions all bids into two groups and runs the monopolistic price auction in each. \citet{lavi2019redesigning} propose using RSOP as a transaction fee mechanism but recognize that it has undesirable properties. Crucially, RSOP has communication complexity $\Omega(n)$ since it requires all bids to be stored on the blockchain.

\subsection{Generalized second-price auction}\label{sec:gsp}
The \emph{generalized second-price} (GSP) auction~\cite{edelman2007internet} orders the bids in decreasing order $b_1 \ge b_2 \ge \ldots \ge b_n$ and allocates slots to the top $m$ bidders with each bidder $i \in [m]$ paying $b_{i+1}$ (assuming $n \geq m+1$). In the context of transaction fee mechanisms, \citet{basu2019towards} proposed a modified generalized second-price auction where bidder $i \in [m]$ pays $b_m$. Note their payment rule is identical to the monopolistic price auction but with a fixed supply. To incentivize miners to not include fake bids and inflate the $m$-th highest bid, the mechanism splits the revenue for block $B_t$ among the creators of future blocks $B_{t + 1}, B_{t+2}, \ldots, B_{t+k}$. As the demand and the number of miners increase, the mechanism is approximately incentive compatible for myopic bidders and myopic miners.

An out-of-model consideration pointed out by \citet{lavi2019redesigning} is collusion between an active miner with bidders where both agree to execute payments outside the blockchain. For example, a bidder bids $b_i = 0$ but agrees to pay the intended bid outside the blockchain.\footnote{For simplicity, we consider the case where the auction has no reserve price.} To motivate bidders to accept the collusion, the miner might divide the revenue with bidders. Moreover, even without colluding with bidders, a miner could create blocks with $m$ fake bids equal to $0$ and share the revenue of blocks $B_{t-k}, B_{t-k+1}, \ldots, B_{t-1}$ for free. The active miner loses $1/k$ of the revenue they could obtain by honestly creating block $B_t$. Thus including only fake transactions in blocks is an approximate equilibrium.

\section{Omitted Proofs from Section~\ref{sec:wdpp}}\label{sec:dpp-appendix}

\begin{proof}[Proof of Theorem~\ref{thm:wdpp}]
Recall $x_i(q)$ is the indicator random variable for bidder $i \in [n] = \{1, \ldots, n\}$ with value $v_i$ receiving a slot when the posted price is $q$. For fixed $\vec v$, $x_i(q)$ is a nonincreasing function of $q$ because the mechanism always allocates slots to the highest bidders. Assuming truthful bidding, the expected value of $\tw(q, B)$ is
\begin{equation}\label{eq:kernel-s}
E_{\tw}(q) := \alpha \underbrace{\frac{1}{m}\e{\sum_{i = 1}^n v_i \cdot x_i(q)}}_{\text{Kernel $\kw(q)$}} + (1-\alpha)q,
\end{equation}
where the expectation is taken over $\vec v$ drawn from $F^n$. From the proof outline in Section~\ref{sec:fixed-point}, we first check $E_{\tw}$ is a monotone mixture. Note that $\kw(q) = \welfare(\vec x(q))/m$ and that $\welfare(\vec x(q))$ is a nonincreasing function of $q$ since $x_i(q) = 1$ for the top $m$ bids with $v_i \geq q$. Thus $E_{\tw}$ is a monotone mixture. Next, we prove the kernel $\kw$ is Lipschitz continuous under the assumption $\rev$ is Lipschitz continuous.
\begin{claim}\label{claim:surplus-lipschitz}
If $\rev$ is $L$-Lipschitz, then kernel $\kw$ is $(L/m+1)$-Lipschitz.
\end{claim}
\begin{proof}
We decompose the welfare as the revenue plus the surplus. Since the revenue is $L$-Lipschitz, it suffices to show the surplus is $m$-Lipschitz. Fix any positive $q < q'$ and because the welfare is decreasing, we have $\kw(q) \geq \kw(q')$. Thus,
\begin{align*}
0 &\leq m \cdot (\kw(q) - \kw(q'))\\
&= \underbrace{\e{\sum_{i = 1}^n (v_i - q) \cdot x_i(q)}}_{\text{Expected surplus at price $q$}} - \underbrace{\e{\sum_{i = 1}^n (v_i - q') \cdot x_i(q')}}_{\text{Expected surplus at price $q'$}} + \rev(q) - \rev(q')\\
&= \e{\sum_{i = 1}^n (v_i - q) \cdot (x_i(q) - x_i(q')) - (v_i - q' - v_i + q) \cdot x_i(q')} +  \rev(q) - \rev(q').
\end{align*}
Observe $x_i(q') \leq x_i(q)$, and whenever $x_i(q) - x_i(q') = 1$, bidder $i$ has value $v_i \leq q'$. Thus the difference between the expected surplus at price $q$ and $q'$ is at most
\begin{align*}
(q' - q)\e{\sum_{i = 1}^n x_i(q) - x_i(q')} + (q' - q)\e{\sum_{i = 1}^n x_i(q')} = (q'-q)\e{\sum_{i = 1}^n x_i(q)} \leq (q'-q)\cdot m. 
\end{align*}
From the assumption $\rev$ is $L$-Lipschitz, we have
$$0 \leq \kw(q) - \kw(q') \leq (q' - q)(L/m + 1).$$
This proves $\kw$ is $(L/m + 1)$-Lipschitz.
\end{proof}
From Lemma~\ref{lemma:contractive-update-rule} with $\alpha \leq 1/(L/m + 1)$, we get that $E_{\tw}(q)$ is a contractive mapping with Lipschitz constant $1-\alpha < 1$. From the Banach Contraction Principle (Lemma~\ref{lemma:banach-contraction}), for any initial positive price $q_0$, the iteration $q_0, E_{\tw}(q_0), E_{\tw}^2(q_0), \ldots$ converges to a unique fixed point $q^*$. From Observation~\ref{obs:equal-fixed point}, $q^*$ is also a unique fixed point for the kernel $\kw$. This proves the intended mechanism is asymptotically stable with respect to distribution $F^n$.
\begin{claim}\label{claim:wdpp-welfare}
If $q$ is a fixed point for $\kw$, then $\welfare(\vec x(q)) \geq \frac{1}{2}\mathit{OPT}$.
\end{claim}
\begin{proof}
Observe $q \cdot m$ is the expected welfare of the WDPP mechanism at equilibrium price $q$ since $q$ is a fixed point for $\kw$. That is,
$$\welfare(\vec x(q)) =  \e{\sum_{i = 1}^n v_i \cdot x_i(q)} = m \cdot \kw(q) = q \cdot m,$$
where the second equality is simply the definition of $\kw$ and the third equality follows from the fact $q = \kw(q)$ is a fixed point of $\kw$. Assume for contradiction $q \cdot m \leq \frac{1}{2} \mathit{OPT}$, then the optimal welfare
\begin{align*}
\mathit{OPT} &= \e{\sum_{i = 1}^n v_i \cdot x_i^*(\vec v)} = \e{\sum_{i = 1}^n v_i(x_i^*(\vec v) - x_i(q))} + \e{\sum_{i = 1}^n v_i \cdot x_i(q)}\\
&<  q \e{\sum_{i = 1}^n (x_i^*(\vec v) - x_i(q))} + \welfare(\vec x(q)) \quad \text{Because $x_i^*(\vec v) - x_i(q) = 1$ only if $v_i < q$.}\\
& \leq \frac{\mathit{OPT}}{2} + \welfare(\vec x(q))\quad \text{Because $q \leq \frac{\mathit{OPT}}{2m}$ and $\sum_{i = 1}^n x_i^*(\vec v) - x_i(q) \leq m$.}\\
&= \frac{\mathit{OPT}}{2}  + \welfare(\vec x(q)).
\end{align*}
Rearranging the inequality gives $q \cdot m = \welfare(\vec x(q)) > \frac{\mathit{OPT}}{2}$, a contradiction. Thus $\welfare(\vec x(q)) = q \cdot m \geq \frac{1}{2}\mathit{OPT}$.
\end{proof}
This proves Theorem~\ref{thm:wdpp}.
\end{proof}

\section{Omitted Proofs from Section~\ref{sec:twdpp}}

\begin{proof}[Proof of Proposition~\ref{prop:twdpp-stable}]
For any $\delta$, $\alpha$, positive value $v$, and supply $m \geq 4(1+\delta)/\delta$, assume that at each time step there are $n = \lfloor m \frac{1+\delta/2}{1+\delta}\rfloor \leq m - 1$ bidders with value $v$. For the UDPP mechanism $\left(\vec x, \vec p, \tu\right)$, if the current price $q_t \leq v$, the subsequent price increases to
$$q_{t+1} = q_t\left(1 + \alpha \left(\frac{n}{m}(1+\delta)-1\right)\right) \geq q_t\left(1+\alpha\left(1+\frac{\delta}{2}-\frac{1+\delta}{m}-1\right)\right) \geq q_t\left(1+\frac{\alpha\delta}{4}\right).$$
If the current price $q_t > v$, no bidder purchases, and the subsequent price decreases to
$$q_{t+1} = q_t(1-\alpha).$$
This repeats {\em ad infinitum}, and the allocation-based update rule is unstable with respect to distribution $F^n$.

For the TWDPP mechanism $\left(\vec x, \vec p, \ttw\right)$, if the current price $q_t > v$, no bidder purchases, and the subsequent price decreases to
$$q_{t+1} = q_t(1-\alpha).$$
If the current price $q_t < \frac{v}{1+\delta}$, all $n$ bidders purchase, and the subsequent price increases to
$$q_{t+1} = \alpha \frac{n}{m}\min\{v, (1+\delta)q_t\} + (1-\alpha)q_t \leq q_t(\alpha(1+\delta)+ 1- \alpha) = q_t(1+\alpha\delta) < v.$$
Thus eventually, the mechanism reaches a price $q_t \in \left[\frac{v}{1+\delta}, v\right] = X$. If the current price $q_t \in X$, all $n$ bidders purchase, and the subsequent price is
$$q_{t+1} = \ttw(q_t) = \alpha \frac{n}{m}\min\{v, (1+\delta)q_t\} + (1-\alpha)q_t = \alpha \frac{n}{m} v + (1-\alpha)q_t \leq v.$$
Additionally,
\begin{align*}
\ttw(q_t) &= \alpha \frac{n}{m} v + (1-\alpha)q_t\\
&\geq \alpha v\left(\frac{1+\delta/2}{1+\delta}-\frac{1}{m}\right)+(1-\alpha)\frac{v}{1+\delta} &  \text{Because $n\geq m\frac{1+\delta/2}{1+\delta}-1$},\\
&= \alpha v\left(\frac{1}{1+\delta} +\frac{\delta/2}{1+\delta}-\frac{1}{m}\right)+(1-\alpha)\frac{v}{1+\delta}\\
& \geq \alpha v\left(\frac{1}{1+\delta} + \frac{\delta}{4(1+\delta)}\right) + (1-\alpha)\frac{v}{1+\delta} &\text{Because $m \geq 4(1+\delta)/\delta$},\\
&\geq \frac{v}{1+\delta}. 
\end{align*}
Thus $\ttw$ maps $X$ to itself. For all $q, q' \in X$, we have
$$|\ttw(q) - \ttw(q')| = (1-\alpha)|q - q'|.$$
Thus the restriction of $\ttw$ to $X$ is a contractive mapping (Definition~\ref{def:contractive-mapping}). From the Banach Contraction Principle (Lemma~\ref{lemma:banach-contraction}), for all $q_t \in X$, the fixed point iteration
$$q_t, \ttw(q_t), \ttw(\ttw(q_t)), \ldots$$
converges to a unique fixed point $z \in X$. This proves the TWDPP mechanism is asymptotically stable with respect to distribution $F^n$. See Figure~\ref{fig:instability-2} for an empirical example of this behavior.
\end{proof}

\subsection{Omitted proofs from Section~\ref{sec:twdpp-stability}}\label{sec:twdpp-stability-appendix}

The following is a well-known fact about concave functions.
\begin{proposition}\label{prop:concave} If $f$ and $g$ are (strictly) concave functions, then:
\begin{enumerate}[label=(\roman*)]
    \item \label{prop-concave-1} $f + g$ is (strictly) concave.
\end{enumerate}
\end{proposition}
We will require the following fact:
\begin{lemma}[Theorem 3.1 in \cite{kennan2001uniqueness}]\label{lemma:concave-uniqueness}
If $f$ is strictly concave and $f(0) \geq 0$, then $f$ has at most one positive fixed point.
\end{lemma}
Lemma~\ref{lemma:concave-uniqueness} guarantees a concave function $f$ satisfying Assumption~\ref{assumption:concave-2} has at most one positive fixed point, but it does not guarantee existence. For the one-dimensional case, uniqueness will follow directly from the intermediate value theorem.
\begin{proof}[Proof of Lemma~\ref{lemma:concave-contraction}]
First, we show $f : X \to X$ with $X = [0, \bar a]$ has at least one nonnegative fixed point.
\begin{claim}\label{claim:concave-1}
There is a positive $a \in X$ such that $f(a) > a$ if and only if $f$ has a unique positive fixed point $x^*$.
\end{claim}
\begin{proof}
We first prove that if there is a positive $a \in X$ such that $f(a) > a$, then $f$ has a unique positive fixed point.

\vspace{1mm}\noindent\textit{$\Rightarrow$} From~\ref{concave-4}, there is a positive $b \in X$ such that $f(b) < b$. Let
\begin{equation}\label{eq:concave-g}
g(x) = f(x) - x,
\end{equation}
and observe $g$ is continuous (because $f$ is continuous) and $x$ is a fixed point for $f$ if and only if $g(x) = 0$. Since $g$ is continuous, $g(a) > 0$, and $g(b) < 0$, we can use the intermediate value theorem to conclude there is an $x^* \in [\min\{a, b\}, \max\{a, b\}] \subseteq X$ such that $g(x^*) = 0$. Thus $x^* > 0$ is a positive fixed point for $f$. From Lemma~\ref{lemma:concave-uniqueness}, $x^*$ is the unique positive fixed point of $f$.

Next, we prove that if $f$ has a unique positive fixed point, then there is a positive $a \in X$ such that $f(a) > a$. 

\vspace{1mm}\noindent\textit{$\Leftarrow$} Let $z \in X$ be a positive fixed point of $f$. For some $\alpha \in (0, 1)$, let $a = \alpha \cdot z > 0$ and note $a \in X$. From the strict concavity of $f$,
\begin{align*}
f(\alpha z) &= f(\alpha z + (1-\alpha) \cdot 0)\\
&> \alpha f(z) + (1-\alpha)f(0) \quad &\text{From strict concavity of $f$,}\\
&\geq \alpha f(z) \quad &\text{From \ref{concave-1}, $f(0) \geq 0$,}\\
&= \alpha z \quad &\text{Because $z = f(z)$ is a fixed point of $f$.}
\end{align*} 
The chain of inequalities witnesses $f(a) = f(\alpha z) > \alpha \cdot z = a > 0$. Thus there is a positive $a \in X$ such that $f(a) > a$.
\end{proof}
\begin{claim}\label{claim:concave-2-existance}
$f$ has at least one nonnegative fixed point and at most one positive fixed point.
\end{claim}
\begin{proof}
Let's first consider the case where $f$ has no positive fixed point. From Claim~\ref{claim:concave-1}, if $f$ has no positive fixed point, then $0 \leq f(x) \leq x$ for all positive $x \in X$. Then the sequence
$$f\left(\bar a\right), f\left(\frac{\bar a}{2}\right), f\left(\frac{\bar a}{3}\right), \ldots$$
converges to $0$. By continuity of $f$, $0 = f(0)$ is a fixed point of $f$ and $0$ is a nonnegative fixed point of $f$ whenever $f$ has no positive fixed point. Second, consider the case where $f$ has a positive fixed point. From Lemma~\ref{lemma:concave-uniqueness}, this fixed point is unique.
\end{proof}
Let $x^*$ be the largest nonnegative fixed point of $f$. Define the partition of $X$ into sets
\begin{equation}\label{eq:A}
A := \{x \in X : f(x) \geq x\},
\end{equation}
\begin{equation}\label{eq:L}
L := \{x \in X : f(x) < x\}.
\end{equation}
\begin{claim}\label{claim:A-L} $A = [\inf X, x^*]$ and $L = (x^*, \sup X]$.
\end{claim}
\begin{proof}
Observe  $x^* \in A$ because $f(x^*) = x^*$. Suppose for contradiction there are $x \in A$ and $y \in L$ such that $x > y$. Because $0 < y < x$, there is an $\alpha \in (0, 1)$ such that $y = (1-\alpha)x$. From the strict concavity of $f$,
\begin{align*}
    f(y) &= f(0\cdot \alpha + x(1-\alpha)) &\\
    &> \alpha f(0) + (1-\alpha)f(x) & \text{From strict concavity of $f$},\\
    &\geq (1-\alpha)x = y & \text{Because $f(0) \geq 0$ and $f(x) \geq x$ since $0, x \in A$}.
\end{align*}
The chain of inequalities witnesses $y \in A$, a contradiction to $y \in L = X \setminus A$.
\end{proof}
We are ready to show that iterating $f$ starting from a positive $x_0$ converges to $x^*$. Define $x_t = f(x_{t-1})$ for $t = 1, 2, \ldots$. It suffices to show that the sequence $x_0, x_1, \ldots$ converges to a limit point $z$ since the continuity of $f$ implies $z$ is a fixed point of $f$. When $z$ is positive and $f$ has a unique positive fixed point, we will have $z = x^*$. First, let us consider the case where $f$ has no positive fixed point.

\begin{claim}\label{claim:concave-2-case-1}
If $f$ has no positive fixed point, the sequence $x_0, x_1, \ldots$ converges to $0$.
\end{claim}
\begin{proof}
From Claim~\ref{claim:concave-2-existance}, $0$ is the unique nonnegative fixed point of $f$ and $L = X \setminus \{0\}$. If $x_0 = 0$, then $f(x_0) = 0$. If $x_0 \in L$, from \eqref{eq:L} and \ref{concave-1}, $0 \leq x_1 = f(x_0) < x_0$. By induction, the sequence
$$x_0 > x_1 > x_2 > \ldots \geq 0.$$
is decreasing and lower bounded by $0$.
\begin{claim}\label{claim:concave-limit-point}
If the sequence $x_0, x_1, \ldots$ is monotone and bounded, the sequence $x_0, x_1, \ldots$ converges to a limit point $z$. Moreover, $z$ is a fixed point of $f$.
\end{claim}
\begin{proof}
Since the Euclidean space is a complete metric space, the sequence $x_0, x_1, \ldots$ converges to a limit point $z$. By continuity of $f$, $z$ is a fixed point of $f$.
\end{proof}
Thus when $f$ has no positive fixed point, the nonincreasing sequence $x_0, x_1, \ldots$ converges to $0$, the unique nonnegative fixed point of $f$.
\end{proof}

For the case where $f$ has a positive fixed point, from Lemma~\ref{lemma:concave-uniqueness}, $f$ has a unique positive fixed point $x^*$. Next, we divide the proof into the case where $x^* < \bar x = \sup_{x \in X} f(x)$ and the case where $x^* \geq \bar x$. 

\begin{claim}\label{claim:concave-2-case-2}
If $f$ has a positive fixed point and $x^* < \bar x$, the sequence $x_t, x_{t+1}, \ldots$ converges to $x^*$ .
\end{claim}
\begin{proof}
Observe $[\inf X, x^*] = A \subseteq I = [\inf X, \bar x)$. We consider separately the case where $x_t \in A$ and the case $x_t \in L$.

\vspace{1mm}\noindent\textbf{Case 1:} For the case $x_t \in A$,
\begin{align*}
x_t &\leq x_{t+1} = f(x_t) & \text{Because $x_t \in A$},\\
&\leq f(x^*) & \text{Because $f$ is increasing on $I \ni x_t, x^*$ and $x_t \leq x^*$},\\
&= x^* & \text{Because $x^* = f(x^*)$ is a fixed point.}
\end{align*}
Thus $x_{t+1} \in A$. By induction, $x_{t + i} \in A$ for all $i \geq 0$. Thus the sequence
$$x_t \leq x_{t+1} \leq \ldots \leq x^*$$
is nondecreasing and upper bounded by $x^*$. From Claim~\ref{claim:concave-limit-point}, the sequence $x_t, x_{t+1}, \ldots$ converges to $x^*$.

\vspace{1mm}\noindent\textbf{Case 2:} For the case where $x_t \in L$, we divide the proof into two subcases. First, for the case where $x_{t+i} \in L = (x^*, \sup X]$ for all $i \geq 0$, the sequence 
$$x_t > x_{t+1} > x_{t+2} > \ldots > x^*$$
is decreasing and lower bounded by $x^*$. From Claim~\ref{claim:concave-limit-point}, the sequence $x_t, x_{t+1}, \ldots$ converges to $x^*$.

Second, for the case where $x_{t+i} \not \in L$ for some $i \geq 0$, $x_{t+i} \in A = X \setminus L$. Thus invoke Case 1 to conclude the sequence $x_{t+i}, x_{t+i+1}, x_{t+i+2}, \ldots$ converges to $x^*$. This proves the sequence $x_t, x_{t+1}, \ldots$ converges to $x^*$ whenever $f$ has a positive fixed point and $x^* < \bar x$.
\end{proof}
\begin{claim}\label{claim:concave-2-case-3}
If $f$ has a positive fixed point and $\bar x \leq x^*$, the sequence $x_t, x_{t+1}, \ldots$ converges to $x^*$.
\end{claim}
\begin{proof}
We first show the restriction of $f$ to $[\bar x, f(\bar x)]$ is a contractive mapping --- i.e., $f$ maps $[\bar x, f(\bar x)]$ to itself and the restriction of $f$ to $[\bar x, f(\bar x)]$ is $L$-Lipschitz with $L < 1$.
\begin{claim}\label{claim:restriction-contraction}
The restriction of $f$ to $[\bar x, f(\bar x)]$ is a contractive mapping.
\end{claim}
\begin{proof}
From assumption~\ref{concave-3}, the restriction of $f$ to the interval $[\bar x, f(\bar x)] \subseteq D$ is $L$-Lipschitz with $L < 1$. Thus it suffices to show that $f$ maps $[\bar x, f(\bar x)]$ to itself. Let $x \in [\bar x, f(\bar x)]$, then $f(x) \leq f(\bar x)$ since $\bar x \leq x$ and the restriction of $f$ to $D$ is decreasing. Suppose for contradiction $f(x) < \bar x$, then
\begin{align*}
|f(x) - f(\bar x)| &= f(\bar x) - f(x)\\
&> f(\bar x) - \bar x \quad &\text{From the assumption $f(x) < \bar x$},\\
&\geq x - \bar x \Rightarrow\Leftarrow \quad &\text{Observe $x \leq f(\bar x)$ since $x \in [\bar x, f(\bar x)]$}.
\end{align*}
We reach a contradiction to the fact that the restriction of $f$ to $D$ is $L$-Lipschitz with $L < 1$. Thus $\bar x \leq f(x) \leq f(\bar x)$, proving $f$ maps $[\bar x, f(\bar x)]$ to itself. This proves the restriction of $f$ to $[\bar x, f(\bar x)]$ is a contractive mapping.
\end{proof}
It follows that if $x_t \in [\bar x, f(\bar x)]$, the sequence $x_t, x_{t+1}, \ldots$ will converge to a limit point.
\begin{claim}\label{claim:concave-2}
For any $x_t \in [\bar x, f(\bar x)]$, the sequence $x_t, x_{t+1}, x_{t+1}, \ldots$ converges to $x^*$.
\end{claim}
\begin{proof}
The proof follows from the Banach Contraction Principle (Lemma~\ref{lemma:banach-contraction}) and the fact the restriction of $f$ to $[\bar x, f(\bar x)]$ is a contractive map (Claim~\ref{claim:restriction-contraction}).
\end{proof}
Recall from Assumption~\ref{assumption:concave-2}, $I = [\inf X, \bar x)$ is the interval where $f$ is nondecreasing and $D = [\bar x, \sup X]$ is the interval $f$ is nonincreasing. Next, we consider separately the case where $x_t \in I$ and the case $x_t \in D$.

\vspace{1mm}\noindent\textbf{Case 1:} For the case $x_t \in I$, observe
$$I = [\inf X, \bar x) \subset [\inf X, x^*] = A,$$
from the assumption $\bar x \leq x^*$. Thus $x_t \in A$. Let's consider two distinct cases depending on whether $x_{t+i} \in A$ for all $i \geq 0$ or $x_{t+i} \not\in A$ for some $i \geq 0$. For the case where $x_{t + i} \in A$ for all $i \geq 0$, the sequence
$$x_t \leq  x_{t+1} \leq x_{t+2} \leq  \ldots \leq \sup_{x \in X} f(x) = f(\bar x)$$
is nondecreasing and upper bounded by $f(\bar x) < \infty$ (because $f$ is Lipschitz continuous):
\begin{claim} $\sup_{x \in X} f(x) < \infty$.
\end{claim}
\begin{proof}
From \ref{concave-4}, there is a $b \in X$ such that $f(b) < b$. Because $f$ is $L$-Lipschitz,
$$\sup_{x \in X}f(x) - f(b) = |f(\bar x) - f(b)| \leq L|\bar x - b| \leq L\sup X.$$
Thus $\sup_{x \in X} f(x) \leq L\sup X + f(b) < L \sup X + b < (L+1)\sup X < \infty$.
\end{proof}
From Claim~\ref{claim:concave-limit-point}, the nondecreasing sequence $x_t, x_{t+1}, \ldots$ converges to $x^*$. For the case where $x_{t+i} \not \in A$ for some $i$, let $i \geq 1$ be the first element that satisfies this condition --- that is, $x_{t+i} \not \in A$ while $x_{t+i-j} \in A$ for all $1 \leq j \leq i$. Observe
$$x_{t+i} = f(x_{t+i-1}) \leq \sup_{x \in X} f(x) = f(\bar x),$$
and the fact $x_{t+i} \not \in A$ implies $x_{t+i} \in X \setminus A = L = (x^*, \sup X]$. Thus $x_{t+i} > x^* \geq \bar x$. Thus
$$\bar x \leq x^* < x_{t+i} = f(x_{t+i-1}) \leq f(\bar x).$$
The chain of inequalities witnesses $x_{t+i} \in [\bar x, f(\bar x)]$ and from Claim~\ref{claim:concave-2}, the sequence $x_{t+i}, x_{t+i+1}, x_{t+i+2}, \ldots$ converges to $x^*$.

\vspace{1mm}\noindent\textbf{Case 2:} For the case where $x_t \in D = [\bar x, \sup X]$, observe $x_t \geq \bar x > 0$. We consider separately the case where $x_{t+i} \in D$ for all $i \geq 0$ and the case where $x_{t+i} \not \in D$ for some $i \geq 0$.

For the case where $x_{t+i} \in D$ for all $i \geq 0$, from Banach contraction principle, the sequence $x_t, x_{t+1}, \ldots$ converges to $x^* \geq \bar x$ since the restriction of $f$ to $D$ is $L$-Lipschitz with $L < 1$~\ref{concave-4}.

For the case where $x_{t+i} \not \in D$ for some $i \geq 0$, $x_{t+i} \in X \setminus D = I$. From Case 1, the sequence $x_{t+i}, x_{t+i+1}, \ldots$ converges to $x^*$. Cases 1 and 2 prove the sequence $x_t, x_{t+1}, \ldots$ converges to $x^*$ whenever $\bar x \leq x^*$ and $f$ has a positive fixed point.
\end{proof}
Claims~\ref{claim:concave-2-case-1}, \ref{claim:concave-2-case-2}, and \ref{claim:concave-2-case-3} together prove that the sequence $x_0, x_1, \ldots$ converges to a fixed point of $f$.
\end{proof}

\begin{proof}[Proof of Lemma~\ref{lemma:reduction}]
Recall that $f : \real_{\geq 0} \to \real_{\geq 0}$ is $L$-Lipschitz, that the restriction of $f$ to $X = [0, \bar x]$ is strictly concave, and that $f(x) = 0$ for all $x \geq \bar a$. For positive $\alpha \leq \frac{1}{L+1}$, define
$$g(x) = \alpha f(x) + (1-\alpha)x.$$
First, we show the restriction of $g : \real_{\geq 0} \to \real_{\geq 0}$ to $X$ maps $X$ to itself and satisfies Assumption~\ref{assumption:concave-2}.

For \ref{concave-0}, observe that $g$ is a convex combination of $x$ and $f(x)$. From Proposition~\ref{prop:concave}, $g$ restricted to $X$ is strictly concave (because $f$ restricted to $X$ is strictly concave).

For \ref{concave-1}, from the assumption $f$ is nonnegative, $g$ is also nonnegative.

For \ref{concave-3}, for all $x \in X$,
\begin{align*}
    0 \leq g(x) &= \alpha f(x) + (1-\alpha)x\\
    &= \alpha (f(x) - f(\bar a)) + (1-\alpha)x \quad &\text{From \ref{concave-2-2}, $f(\bar a) = 0$},\\
    &\leq \alpha (\bar a - x)L + (1-\alpha)x \quad  &\text{From the fact $f$ is $L$-Lipschitz},\\
    &\leq \bar a - x + x - \alpha x\quad &\text{From the fact $\alpha \leq 1/L$},\\
    &\leq \bar a.
\end{align*}
Thus $g$ maps $X$ to itself. Let $I \subseteq X$ be the interval where $f$ is increasing and let $D \subseteq X$ be the interval where $f$ is decreasing. Then the restriction of $g$ to $D$ is a monotone mixture (Definition~\ref{def:monotone-mixture}) with kernel $f$. From Lemma~\ref{lemma:contractive-update-rule}, the restriction of $g$ to $D$ is $(1-\alpha)$-Lipschitz for $\alpha \leq \frac{1}{L+1}$. This proves the restriction of $g$ to $X$ satisfies \ref{concave-3}.

For \ref{concave-4}, let $b = \bar a > 0$ and from the assumption $f(b) = 0 < b$,
$$g(b) = \alpha f(b) + (1-\alpha)b < b.$$
The work above proves the restriction of $g$ to $X$ satisfies Assumption~\ref{assumption:concave-2}. Consider the fixed-point iteration of $g$ starting from a positive $x_0$. Let $x_{t+1} = g(x_t)$ for $t \geq 0$. If $x_t \geq \bar a$, then
$$g(x_t) = (1-\alpha) x_t.$$
Let $x_t$ be the first element in the sequence $x_0, x_1, \ldots$ such that $x_t \leq \bar a$. From Lemma~\ref{lemma:concave-contraction}, the fixed-point iteration of $g$ restricted to $X$ starting from $x_t \in X$ converges to a nonnegative fixed point of $g$. Finally, observe that a fixed point of $g$ is also a fixed point of $f$.
\end{proof}

\begin{example}\label{example:reduction}
Consider the strictly concave functions $f_1(x) = 4 - (x-2)^2$ and $f_2(x) = \frac{f(x)}{2}$, and for $\alpha = 0.4$, define functions $g_1(x) = \alpha f_1(x) + (1-\alpha) x$ and $g_2(x) = \alpha f_2(x) + (1-\alpha)x$. In Figure~\ref{fig:concave-reduction}, we plot the fixed-point iteration of $g_1$ (respectively $g_2$), demonstrating that for a function satisfying Assumption~\ref{assumption:concave-1} (e.g., $f_1$ (resp. $f_2$)), we can construct a function satisfying Assumption~\ref{assumption:concave-2} (e.g., $g_1$ (resp. $g_2$)). Then the fixed-point iteration of $g_1$ (resp. $g_2$) converges to a fixed point for $f_1$ (resp. $f_2$).
\begin{figure}
    \centering
    \includegraphics[scale=0.90]{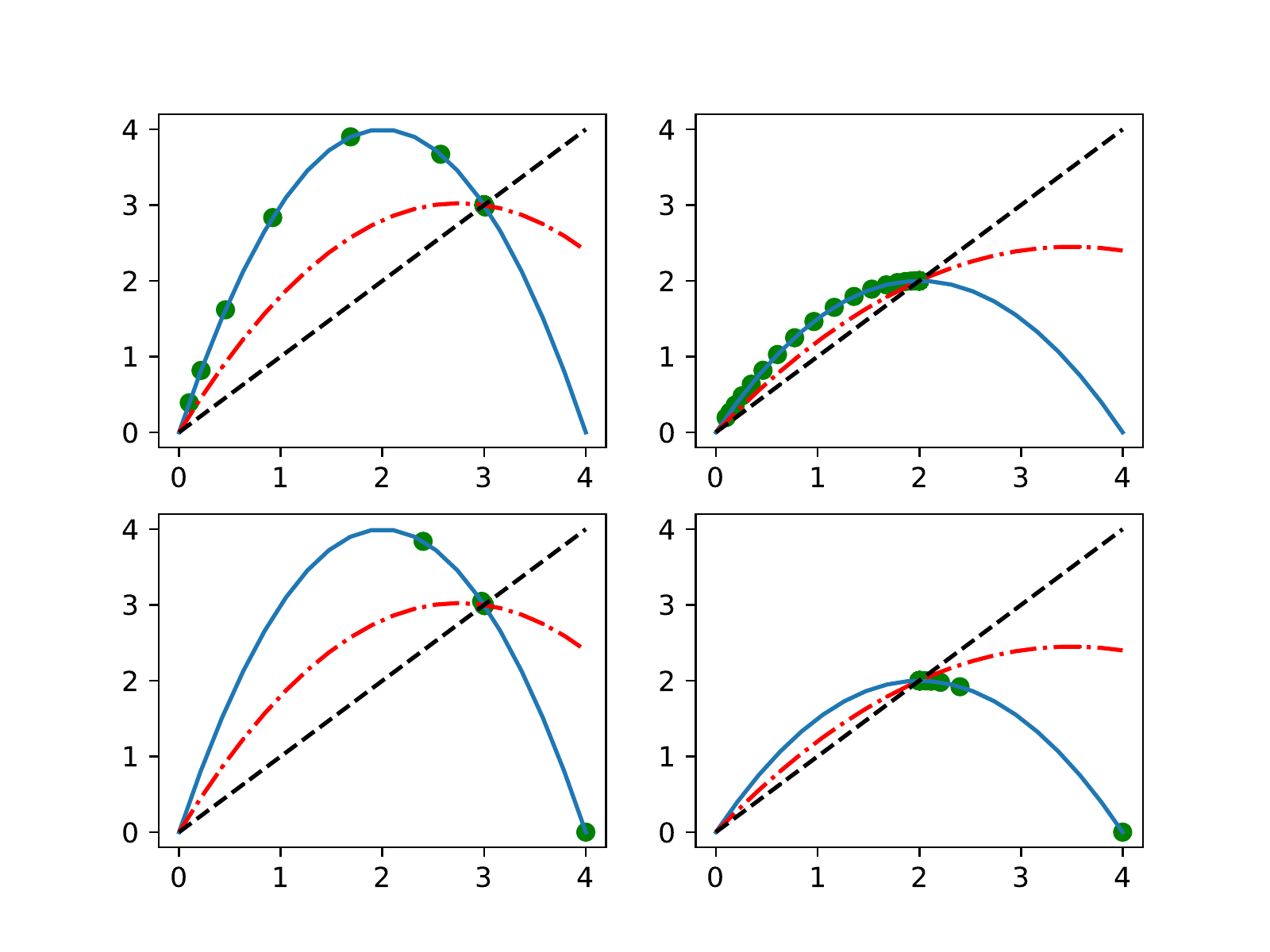}
    \caption{(Top-left) $g_1$ with initial point $q_0 = 0.1$; (Top-right) $g_2$ with initial point $q_0 = 0.1$; (Bottom-left) $g_1$ with initial point $q_0 = 4$; (Bottom-right) $g_2$ with initial point $q_0 = 4$. The solid line (blue) denotes function $f_1$ or $f_2$. The dashed-dot line (red) denotes function $g_1$ or $g_2$. The dashed line (black) is the function $y(x) = x$. Dots (green) denote the following sequence of points in $\mathbb R^2$: $(q_0, f(q_0)), (f(q_0), f^2(q_0)), (f^2(q_0), f^3(q_0)), \ldots$.}
    \label{fig:concave-reduction}
\end{figure}
\end{example}

\subsection{Omitted proofs from Section~\ref{sec:twdpp-welfare}}\label{sec:twdpp-welfare-appendix}

\begin{proposition}\label{prop:rm-welfare}
Consider the RDPP mechanism $(\vec x, \vec p, T)$ with equilibrium price $q$. If $\welfare(\vec x(q)) \geq \frac{q \cdot m}{1+\delta}$ and $Pr[N(q) < m] \geq \frac{\delta}{1+\delta}$, then 
$$\welfare(\vec x(q)) \geq \frac{\mathit{OPT}}{2(1+\delta)}\min \lbrace 1, \delta\rbrace.$$
\end{proposition}
\begin{proof}
Let $c$ be a positive constant. First, we consider two cases:

\vspace{1mm}\noindent\textbf{Case 1.} For the case $q \geq \frac{c}{m} OPT$, from the assumption $\welfare(\vec x(q)) \geq \frac{q \cdot m}{1+\delta}$, we have
\begin{equation}\label{eq:w-case-1}
\welfare(\vec x(q)) \geq \frac{q \cdot m}{1+\delta} \geq \frac{c}{1+\delta} OPT.
\end{equation}

\vspace{1mm}\noindent\textbf{Case 2.} Next, consider the case $q < \frac{c}{m} OPT$. From the assumption $\pr{N(q) < m} \geq \frac{\delta}{1+\delta}$, we might hope to not lose too much welfare in the event $\{N(q) \geq m\}$ since this event is correlated with bidders having smaller values. Moreover, it is easy to estimate the welfare whenever $N(q) \leq m$ --- it is simply the sum of all $v_i \geq q$. Since $q < \frac{c}{m} OPT$, we can upper bound the welfare loss from the top $N(q) - m$ bidders with values $v_i < q$. Let $N_{-i}(q) := \sum_{j \neq i} \mathbbm 1_{v_j \geq q}$ and let $\vec v_{-i} = (v_1, \ldots, v_{i-1}, v_{i+1}, \ldots, v_n)$ be the vector $\vec v$ excluding the $i$-th entry. Then,
\begin{align*}
\welfare(\vec x(q)) &= \e{\sum_{i = 1}^n v_i \cdot x_i(q) \cdot \mathbbm 1_{N(q) \leq m}} + \e{\sum_{i = 1}^n v_i \cdot x_i(q) \cdot \ind{N(q) > m}}\\
&\geq \e{\sum_{i = 1}^n v_i \cdot x_i(q) \cdot \ind{N(q) \leq m}}\\
&= \e{\sum_{i = 1}^n v_i \cdot x_i^*(\vec v) \cdot \ind{v_i \geq q} \cdot \ind{N(q) \leq m}} \\
&= \e{\sum_{i = 1}^n v_i \cdot x_i^*(\vec v) \cdot \ind{v_i \geq q} \cdot \ind{N_{-i}(q) \leq m-1}} \\
&= \sum_{i = 1}^n \mathbb E_{\vec v_{-i} \sim F^{n-1}}\left[\ind{N_{-i}(q) \leq m-1} \cdot \mathbb E_{v_i \sim F}\left[v_i \cdot x_i^*(\vec v) \cdot \ind{v_i \geq q} \cdot \ind{N_{-i}(q) \leq m-1} | \vec v_{-i}\right] \right].
\end{align*}
The third line observes that given $N(q) \leq m$, we have that $x_i(q) = 1$ if and only if bidder $i$ is one of the top $m$ bidders with value $v_i \geq q$. The fourth line observes that the following events are equivalent:
\begin{itemize}
\item Bidder $i$ with value $v_i \geq q$ is one of the top $m$ bidders when there are at most $m$ bidders with value at least $q$.
\item Bidder $i$ with value $v_i \geq q$ is one of the top $m$ bidders when among bidders $[n]\setminus \{i\}$ there are at most $m-1$ bidders with value at least $q$.
\end{itemize}
The fifth line follows from linearity of expectation and the law of total expectation by observing that $\vec v$ is drawn from the product distribution $F^n = \underbrace{F \times \ldots \times F}_{n\text{ times}}$. Next, observe that if among bidders $[n]\setminus \{i\}$ there are at most $m-1$ bidders with value at least $q$, then the event that bidder $i$ has value at least $q$ implies bidder $i$ is one of the top $m$ bidders. Thus $\welfare(\vec x(q))$ is at least
\begin{align*}
&\geq \sum_{i = 1}^n \mathbb E_{\vec v_{-i} \sim F^{n-1}}\left[\ind{N_{-i}(q) \leq m-1} \cdot \mathbb E_{v_i \sim F}\left[v_i \mathbbm \cdot \ind{v_i \geq q} | \vec v_{-i}\right] \right]\\
&= \sum_{i = 1}^n \pr{N_{-i}(q) \leq m-1} \e{v_i \cdot \ind{v_i \geq q}}\quad \text{From independence between $v_i$ and $\vec v_{-i}$,}\\
&\geq \pr{N(q) < m}  \e{\sum_{i = 1}^n v_i \cdot \ind{v_i \geq q}} \quad \text{$N(q) < m$ implies $N_{-i}(q) \leq m-1$,}\\
&\geq \frac{\delta}{1+\delta} \e{\sum_{i = 1}^nv_i \mathbbm \cdot \ind{v_i \geq q} } \quad \text{From the assumption $N(q) < m$ with probability at least $\delta/(1+\delta)$,}\\
&=\frac{\delta}{1+\delta} \left(\e{\sum_{i = 1}^n v_i \cdot x_i^*(\vec v)} + \e{\sum_{i = 1}^n v_i(\ind{v_i \geq q} - x_i^*(\vec v))}\right)\\
&\geq \frac{\delta}{1+\delta} \left(\e{\sum_{i = 1}^n v_i \cdot x_i^*(\vec v)} - \e{\sum_{i = 1}^n v_i \cdot x_i^*(\vec v) \cdot \ind{v_i < q}}\right)\\
&> \frac{\delta}{1+\delta}\left(OPT - q \cdot m\right) \quad \text{From the fact $v_i \cdot \ind{v_i < q} < q$ and $\sum_{i = 1}^n x_i^*(\vec v) \leq m$,}\\
&\geq \frac{\delta}{1+\delta}(1-c)OPT \quad \text{From the assumption $q \cdot m < c \cdot OPT$.}
\end{align*}
The second line observes that for independent random variables $X$ and $Y$, $\e{X\cdot Y} = \e{X}\e{Y}$. The chain of inequalities proves that for Case 2,
$$\welfare(\vec x(q)) \geq \frac{\delta}{1+\delta}(1-c)OPT.$$
Combining Case 1 and 2, we get that
$$\welfare(\vec x(q)) \geq \frac{OPT}{1+\delta}\min \lbrace c, (1-c)\delta\rbrace.$$
Setting $c = 1/2$ proves Proposition~\ref{prop:rm-welfare}.
\end{proof}

\begin{proof}[Proof of Theorem~\ref{thm:twdpp-welfare}]
Let $N_{-i}(q) = \sum_{j \neq i} \ind{v_j \geq q}$ be the number of bidders that would purchase a slot at price $q$ excluding bidder $i$. Recall that the value profile $\vec v$ is drawn from the product distribution $F^n$ and whenever the demand $N(q) > m$, the active miner selects $m$ random bidders with values at least $q$. Recall that the expected value of the random operator $\ttw(q, B)$ is
$$E_{\ttw}(q) = \alpha \ktw(q) + (1-\alpha)q,$$
and if $q$ is an equilibrium price, $q = E_{\ttw}(q)$ is a fixed point of $E_{\ttw}$. From the expression for $E_{\ttw}$, we get that $q = \ktw(q)$ is also a fixed point of $\ktw$. By the definition of $\ktw$ (Equation~\ref{eq:ktw}), if $q$ is a fixed point of $\ktw$, we have
\begin{equation}\label{eq:prob-full}
q = \ktw(q) \geq (1+\delta)q Pr[N(q) \geq m].
\end{equation}
Thus, the probability that at price $q$, the demand is at most $m-1$ is
\begin{equation}\label{eq:prob-clearing}
\pr{N(q) < m} =  1 - \pr{N(q) \geq m} \geq \frac{\delta}{1+\delta}.
\end{equation}
Next, we lower bound  $\welfare(\vec x(q))$ in terms of $q$ and $\delta$. Because $q$ is a fixed point of $\ktw$, we have
\begin{align*}
q \cdot m &= \ktw(q) \cdot m = \e{\sum_{i = 1}^n \min\{v_i, (1+\delta)q\} \cdot x_i(q) \cdot \ind{N(q) < m} + (1+\delta) \cdot q \cdot \ind{N(q) \geq m}}\\
& \leq \e{\sum_{i = 1}^n v_i \cdot x_i(q) \cdot \ind{N(q) < m} + (1+\delta) \cdot q \cdot \ind{N(q) \geq m}}\\
& = \welfare(\vec x(q)) + \e{\sum_{i = 1}^n ((1+\delta) \cdot q - v_i)\cdot x_i(q) \cdot \ind{N(q) \geq m}}\\
&\leq \welfare(\vec x(q)) + \delta \cdot q \cdot \e{\sum_{i = 1}^n x_i(q) \cdot \ind{N(q) \geq m} } \quad \text{Observing $x_i(q) = 1$ only if $v_i \geq q$.}\\
&= \welfare(\vec x(q)) + \delta \cdot q \cdot m \cdot \pr{N(q) \geq m}\\
&\leq \welfare(\vec x(q)) + \frac{\delta \cdot q \cdot m}{1+\delta} \quad \text{From \eqref{eq:prob-full}.}
\end{align*}

Rearranging the inequality, we get
\begin{equation}\label{eq:w-welfare-bound}
\welfare(\vec x(q)) \geq \frac{q \cdot m}{1+\delta}.
\end{equation}
From \eqref{eq:prob-clearing}, \eqref{eq:w-welfare-bound},  and Proposition~\ref{prop:rm-welfare}, we establish
$$\welfare(\vec x(q)) \geq \frac{OPT}{2(1+\delta)}\min\{1, \delta\}.$$
\end{proof}

\end{document}